%% file: cyclic_pdl.tex
\documentclass{llncs}

\usepackage{amsmath}
\usepackage{amssymb}
\usepackage[
    colorlinks=true
    , linkcolor=blue
    , urlcolor=blue
    , citecolor=blue
  ]{hyperref} 
\usepackage{cleveref}
\usepackage[inline]{enumitem}
\usepackage[
      nomargin
    , inline
  ]{fixme}
\usepackage{graphicx}
\usepackage{mathtools}
\usepackage{prooftree}
\usepackage{relsize}
\usepackage{subfig}
\usepackage{tabularx}
\usepackage{thm-restate}
\usepackage{thmtools}
\usepackage{tikz}
\usepackage{xcolor}

\usetikzlibrary{arrows}
\usetikzlibrary{shapes.arrows}
\usetikzlibrary{calc}
\usetikzlibrary{positioning}

\tikzstyle{arrow}+=[thick,rounded corners=0.5em]
\tikzstyle{every picture}+=[remember picture,baseline]

\fxusetheme{color}
\fxuseenvlayout{color}

\FXRegisterAuthor{rr}{arr}{\color{teal}[reuben]}
\FXRegisterAuthor{sd}{asd}{\color{magenta}[simon]}

\newlist{propertyenum}{enumerate}{1}
\newlist{conditionenum}{enumerate}{1}
\newlist{lemmaenum}{enumerate}{1}
\newlist{lemmaenum*}{enumerate*}{1}
\setlist[propertyenum,conditionenum,lemmaenum,lemmaenum*]{label={(\arabic*)}}

\crefname{axiom}{axiom}{axioms}
\Crefname{axiom}{Axiom}{Axioms}
\creflabelformat{axiom}{(#2#1#3)}
\crefrangelabelformat{axiom}{(#3#1#4) to~(#5#2#6)}

\crefname{cond}{condition}{conditions}
\Crefname{cond}{Condition}{Conditions}
\creflabelformat{cond}{(#2#1#3)}
\crefrangelabelformat{cond}{(#3#1#4) to~(#5#2#6)}

\crefname{property}{property}{properties}
\Crefname{property}{Property}{Properties}
\creflabelformat{property}{(#2#1#3)}
\crefrangelabelformat{property}{(#3#1#4) to~(#5#2#6)}

\crefname{invariant}{invariant}{invariants}
\Crefname{invariant}{Invariant}{Invariants}
\creflabelformat{invariant}{(#2#1#3)}
\crefrangelabelformat{invariant}{(#3#1#4) to~(#5#2#6)}

\crefname{propertyenumi}{property}{properties}
\Crefname{propertyenumi}{Property}{Properties}
\creflabelformat{propertyenumi}{(#2#1#3)}
\crefrangelabelformat{propertyenumi}{(#3#1#4) to~(#5#2#6)}

\crefname{conditionenumi}{condition}{conditions}
\Crefname{conditionenumi}{Condition}{Conditions}
\creflabelformat{conditionenumi}{(#2#1#3)}
\crefrangelabelformat{conditionenumi}{(#3#1#4) to~(#5#2#6)}

\crefname{lemmaenumi}{lemma}{lemmas}
\Crefname{lemmaenumi}{Lemma}{Lemmas}
\creflabelformat{lemmaenumi}{#2#1#3}
\crefrangelabelformat{lemmaenumi}{#3#1#4 to~#5#2#6}

\crefname{lemmaenum*i}{lemma}{lemmas}
\Crefname{lemmaenum*i}{Lemma}{Lemmas}
\creflabelformat{lemmaenum*i}{#2#1#3}
\crefrangelabelformat{lemmaenum*i}{#3#1#4 to~#5#2#6}

\newcommand{\choice}[2]{\ensuremath{{#1}\mathbin{\choiceSymb}{#2}}}
\newcommand{\choiceSymb}{\ensuremath{{\cup}}}
\newcommand{\compSymb}{\ensuremath{;}}
\newcommand{\comp}[2]{\ensuremath{{#1}\mathbin{\compSymb}{#2}}}
\newcommand{\emptysequence}{\ensuremath{\varepsilon}}
\newcommand{\falsum}{\ensuremath{\bot}}
\newcommand{\InfinitarySystem}{\ensuremath{\mathbf{G3PDL}^{\infty}}}

\newcommand{\counterEgs}[3]{\ensuremath{\mathcal{C}_{({#1},{#2})}({#3})}}
\newcommand{\CyclicSystem}{\ensuremath{\mathbf{G3PDL}^{\omega}}}
\newcommand{\ID}[1]{\ensuremath{\mathsf{Id}({#1})}}
\newcommand{\idValuation}{\ensuremath{\mathfrak{v}}}
\newcommand{\interp}[2][\model]{\ensuremath{\interpSym[{#1}]({#2})}}
\newcommand{\interpSym}[1][\model]{\mathcal{I}_{#1}}
\newcommand{\iter}[1]{\ensuremath{{#1}^{\iterSymb}}}
\newcommand{\iterSymb}{\ensuremath{\ast}}
\newcommand{\Labels}{\ensuremath{\mathcal{L}}}
\newcommand{\labs}[1]{\ensuremath{\mathsf{labs}(#1)}}
\newcommand{\measureOf}[3]{\ensuremath{\mu_{({#1},{#2})}({#3})}}

\newcommand{\model}{\ensuremath{\mathfrak{m}}}
\newcommand{\nec}[2]{\ensuremath{[{#1}]{#2}}}

\newcommand{\pathsetMeasure}[1]{\ensuremath{\mu({#1})}}
\newcommand{\prefix}{\ensuremath{\sqsubseteq}}
\newcommand{\relatedBy}[1]{\ensuremath{\mathrel{R_{#1}}}}
\newcommand{\sequent}[2]{\ensuremath{{#1} \Rightarrow {#2}}}
\newcommand{\sizeOf}[3][{(\model, \valuation)}]{\ensuremath{\mu_{#1}({#2},{#3})}}
\newcommand{\starredlabs}[1]{\ensuremath{\mathsf{\ast\text{-}labs}(#1)}}
\newcommand{\subst}[2]{\ensuremath{\{{#1}/{#2}\}}}

\newcommand{\test}[1]{\ensuremath{{#1}\testSymb}}
\newcommand{\testSymb}{\ensuremath{?}}
\newcommand{\truncateAt}[2]{\ensuremath{{\lfloor {#2} \rfloor}_{#1}}}
\newcommand{\valuation}{\ensuremath{v}}

\DeclareMathOperator{\combinationsFor}{\ensuremath{!}}
\DeclareMathOperator{\languageOf}{\mathcal{L}}
\DeclareMathOperator{\pathMax}{\mathsf{path}\text{-}\mathsf{max}}
\DeclareMathOperator{\starMax}{\ast\text{-}\mathsf{max}}
\DeclareMathOperator{\unfold}{\Lambda}

\allowdisplaybreaks

\begin{document}

\title{A Non-wellfounded, Labelled Proof System for Propositional Dynamic Logic}
\author{Simon Docherty\inst{1} and Reuben N. S. Rowe\inst{2}}
\institute{
  Department of Computer Science,
  University College London, UK
  \texttt{simon.docherty@ucl.ac.uk}
    \and
  School of Computing,
  University of Kent,
  Canterbury, UK
  \texttt{reuben.rowe@kent.ac.uk}
}

\maketitle

\begin{abstract}
  We define an infinitary labelled sequent calculus for PDL, {\InfinitarySystem}.
  A finitarily representable cyclic system, {\CyclicSystem}, is then given.
  We show that both are sound and complete with respect to standard models of PDL and, further, that {\InfinitarySystem} is cut-free complete.
  We additionally investigate proof-search strategies in the cyclic system for the fragment of PDL without tests.
\end{abstract}

\section{Introduction}

Fischer and Ladner's Propositional Dynamic Logic (PDL) \cite{FischerLadner1979}, which is the propositional variant of Pratt's Dynamic Logic \cite{Pratt1976}, is perhaps \emph{the} quintessential modal logic of action. 
While (P)DL arose initially as a modal logic for reasoning about program execution 
its impact as a formalism for extending `static' logical systems with `dynamics'
via composite actions \cite[p. 498]{Harel1984} has been felt broadly across logic. This is witnessed 
in extensions and variants designed for reasoning about games \cite{Parikh1983}, 
natural language \cite{Groenendijk1991}, cyber-physical systems \cite{Platzer2008}, epistemic agents \cite{GSL2012}, XML \cite{Afanasiev2005}, 
and knowledge representation \cite{DeGiacomo1994}, among others. 

What is lacking, however, is a uniform proof theory for PDL-type logics. 
Much of the proof theoretic work on these logics focuses on Hilbert-style axiomatisations, 
which are not amenable to automation.
%
Proof systems for PDL itself can broadly be characterised as one of two sorts. 
Falling into the first category are a multitude of infinitary systems \cite{RenardeldeLavalette2008,Hill2010,Frittella2014}
employing either infinitely-wide $\omega$-proof rules, or (equivalently) allowing countably infinite contexts.
In the other category are tableau-based algorithms for deciding PDL-satisfiability \cite{DeGiacomo2000,GoreWidmann2009}.
While these are (neccessarily) finitary, they require a great deal of auxillary structure tailored to the decision procedure for PDL itself, and it is unclear how they might be generalised to systems for variants or extensions of the logic.


%

In the proof theory of modal logic, a high degree of uniformity and modularity 
has been achieved through labelled systems. The idea of using 
labels as syntactic representatives of Kripke models in 
modal logic proof systems can be traced back to Fitting \cite{Fitting1983}, 
and a succinct history of the use of labelled systems is provided by Negri \cite{Negri2011}. 
Negri's work \cite{Negri2005} is the high point of the technique, giving a 
procedure to transform frame conditions for Kripke models 
into labelled sequent calculi rules preserving structural properties 
of the proof system, given they are defined as \emph{coherent axioms}. 

The power of this rule generation technique is of particular interest 
because it enables the specification of sound and complete systems 
for classes of Kripke frames that are first-order, but not modally, 
definable. In the context of PDL-type logics, this is of interest because 
of common additional program constructs like intersection which 
have a non-modally definable intended interpretation \cite{Passy1991}. However, even 
with this expressive power, such a framework on its own 
cannot account for program modalities involving iteration. 
In short, formulae involving these modalities are interpreted 
via the reflexive-transitive closure of accessibility relations, 
and this closure is not first-order (and therefore, not coherently) 
definable. Something more must be done to capture the PDL family of logics. 

In this paper we provide the first step towards a uniform proof theory 
of the sort that is currently missing for this family of logics 
by giving two new proof systems for PDL. We combine two ingredients 
from modern proof theory that have hitherto remained separate: 
labelled deduction \`{a} la Negri and non-wellfounded and \emph{cyclic} sequent calculi. 

We first construct a labelled sequent calculus {\InfinitarySystem},
extending that of Negri \cite{Negri2005}, in which proofs are permitted 
to be infinitely tall. For this system soundness (via descending counter-models) and cut-free completeness (via counter-model construction)
are proved in a similar manner to Brotherston and Simpson's 
infinitary proof theory for first-order logic with inductive definitions \cite{BrotherstonSimpson11}. 
Next we restrict attention to \emph{regular} proofs, meaning only those infinite proof trees 
that are finitely representable (i.e.~only have a finite number of distinct subtrees), obtaining the cyclic system \CyclicSystem. 
This can be done by permitting the forming of backlinks (or, cycles) in the 
proof tree, granted a (decidable) trace condition guaranteeing 
soundness can be established. We then show that the
axiomatisation of PDL \cite{HKT2000} can be derived in \CyclicSystem, obtaining completeness.  
We finish the paper with an investigation of proof-search in the cyclic system for
a sub-class of sequents, and conjecture cut-free completeness for the test-free fragment of PDL.

Most crucially, this gives a simple (finitary) sequent calculus that 
elegantly handles iteration in a fashion reflecting actual reasoning 
about the operation. 
We conjecture that this system can be adapted to uniformly handle, 
not just additional program constructs that are found in the PDL literature, 
but also other modal logics (for example, epistemic logics with common knowledge 
modalities) whose interpretation requires the transitive closure of 
accessibility relations, which are perhaps also defined by 
coherent axioms. We discuss this, and a plan for 
future work 
in the conclusion. 

For space reasons we elide proofs, but include them in an appendix.

\paragraph*{Related Work.}

Beyond the proof systems outlined above, the most significant 
related work can be found in Das and Pous' \cite{DasPous17,DasPous18} cyclic proof systems 
for deciding Kleene algebra (in)equalities. Das and Pous' insight 
that iteration can be handled in a cyclic sequent calculus is 
essential to our work here, although there are additional complexities 
involved in formulating a system for PDL because of the interaction 
between programs (which form a Kleene algebra with tests) and 
formulae. We also note that Gor\'{e} and Widmann's tableau procedure 
also utilises the formation of cycles in proof trees. Our proof of cut-free completeness of the
 infinitary system also follows that of Brotherston and Simpson \cite{BrotherstonSimpson11} for first-order logic
with inductive definitions.

Recent work by Cohen and Rowe \cite{CohenRowe18} gives a 
cyclic proof system for the extension of first-order logic with 
a transitive closure operator and we conjecture 
that our labelled cyclic system (and labelled cyclic systems for modal 
logics more generally) can be formalised within it. This idea 
echoes van Benthem's suggestion that the most natural
setting for many modal logics is not first-order logic, but in 
fact first-order logic with a least fixed point operator \cite{Benthem2006}.

Cyclic proof systems have also been defined for some modal logics 
with similar model properties to PDL, including the logic of 
common knowledge \cite{Wehbe2010}  and G\"{o}del-L\"{o}b logic \cite{Shamkanov2014}. 
The idea of cyclic proof can be traced to 
modal $\mu$-calculus \cite{Niwinksi1996}.
Indeed, it can be shown that the logic of common knowledge \cite{Alberucci2002}, G\"{o}del-L\"{o}b logic \cite{Benthem2006,Visser2005} and PDL \cite{Benthem2006,Carreiro2014} can be faithfully interpreted in the modal $\mu$-calculus, indicating that 
perhaps cyclic proof was the right approach for PDL all along.

\section{PDL: Syntax and Semantics}

The syntax of PDL formulas is defined as follows.
We assume countably many atomic \emph{propositions} (ranged over by $p$, $q$, $r$), and countably many atomic \emph{programs} (ranged over by $a$, $b$, $c$).

\begin{definition}[Syntax of PDL]
  The set of \emph{formulas} ($\varphi$, $\psi$, $\ldots$) and the set of \emph{programs} ($\alpha$, $\beta$, $\ldots$) are defined mutually by the following grammar:
  \begin{align*}
    \varphi, \psi
      &\Coloneqq 
             \falsum
        \mid p
        \mid \varphi \wedge \psi
        \mid \varphi \vee \psi
        \mid \varphi \rightarrow \psi
        \mid \nec{\alpha}{\varphi}
      \\
    \alpha, \beta, \gamma
      &\Coloneqq
             a
        \mid \comp{\alpha}{\beta}
        \mid \choice{\alpha}{\beta}
        \mid \test{\varphi}
        \mid \iter{\alpha}
  \end{align*}
\end{definition}

We briefly reprise the semantics of PDL (see \cite[\textsection{5.2}]{HKT2000}).
A PDL model $\model = (\mathcal{S}, \mathcal{I})$ is a Kripke model consisting of a set $\mathcal{S}$ of \emph{states} and an \emph{interpretation} function $\mathcal{I}$ that assigns: a subset of $\mathcal{S}$ to each atomic proposition; and a binary relation on $\mathcal{S}$ to each atomic program.
We inductively construct an extension of the interpretation function, denoted $\interpSym$, that operates on the full set of propositions and programs.

\begin{definition}[Semantics of PDL]
\label{def:Semantics}
  Let $\model = (\mathcal{S}, \mathcal{I})$ be a PDL model. We define the extended interpretation function $\interpSym$ inductively as follows:
  \begin{align*}
    \interp{\bot} &= \emptyset
      &
    \interp{a} &= \mathcal{I}(a)
      \\
    \interp{p} &= \mathcal{I}(p)
      &
    \interp{\comp{\alpha}{\beta}} &= \interp{\alpha} \circ \interp{\beta}
      \\
    \interp{\varphi \wedge \psi} &= \interp{\varphi} \cap \interp{\psi} 
      &
    \interp{\choice{\alpha}{\beta}} &= \interp{\alpha} \cup \interp{\beta}
      \\
    \interp{\varphi \vee \psi} &= \interp{\varphi} \cup \interp{\psi} 
      &
    \interp{\test{\varphi}} &= \ID{\interp{\varphi}}
      \\
    \interp{\varphi \rightarrow \psi} &= (\mathcal{S} \setminus \interp{\varphi}) \cup \interp{\psi}
      &
    \interp{\iter{\alpha}} &= \bigcup_{k \geq 0} {\interp{\alpha}^{k}}
      \\[-0.75em]
    \interp{\nec{\alpha}{\varphi}} &= \mathcal{S} \setminus \Pi_1(\interp{\alpha} \circ \ID{\mathcal{S} \setminus \interp{\varphi}})
  \end{align*}
  where $\circ$ denotes relational composition, $R^{n}$ denotes the composition of $R$ with itself $n$ times, $\Pi_1$ returns a set by projecting the first component of each tuple in a relation, and $\ID{X}$ denotes the identity relation over the set $X$.
\end{definition}

We write $\model, s \models \varphi$ to mean $s \in \interp{\varphi}$, and $\model \models \varphi$ to mean that $\model, s \models \varphi$ for all states $s \in \mathcal{S}$.
A PDL formula $\varphi$ is \emph{valid} when $\model \models \varphi$ for all models $\model$.

\section{An Infinitary, Labelled Sequent Calculus}

We now define a sequent calculus for deriving theorems (i.e.~valid formulas) of PDL.
This proof system has two important features.
The first is that it is a \emph{labelled} proof system.
Thus sequents contain assertions about the structure of the underlying Kripke models and formulas are labelled with atoms denoting specific states in which they should be interpreted.
Secondly, we allow proofs of infinite height.

We assume a countable set $\Labels$ of \emph{labels} (ranged over by $x$, $y$, $z$) that we will use to denote particular states.
A \emph{relational atom} is an expression of the form \mbox{$x \relatedBy{a} y$}, where $x$ and $y$ are labels and $a$ is an atomic program.
A \emph{labelled formula} is an expression of the form $x : \varphi$, where $x$ is a label and $\varphi$ is a formula.
We define a label substitution operation by $z \subst{x}{y} = y$ when $x = z$, and $z \subst{x}{y} = z$ otherwise.
We lift this to relational atoms and labelled formulas by: $({z \relatedBy{a} z'}) \subst{z}{y} = {x \subst{z}{y}} \relatedBy{a} {z' \subst{x}{y}}$ and $({z : \varphi}) \subst{x}{y} = {{z \subst{x}{y}} : \varphi}$.

Sequents are expressions of the form $\sequent{\Gamma}{\Delta}$, where $\Gamma$ and $\Delta$ are finite sets of relational atoms and labelled formulas.
We denote an arbitrary member of such a set using $A$, $B$, etc.
As usual, $\Gamma, A$ and $A, \Gamma$ both denote the set $\{ A \} \cup \Gamma$, and $\Gamma \subst{z}{y}$ denotes the application of the (label) substitution $\subst{x}{y}$ to all the elements in $\Gamma$.
We denote by $\nec{\alpha}{\Gamma}$ the set of formulas obtained from $\Gamma$ by prepending the modality $\nec{\alpha}{}$ to every labelled formula.
That is, we define $\nec{\alpha}{\Gamma} = \{ x \relatedBy{a} y \mid x \relatedBy{a} y \in \Gamma \} \cup \{ x : \nec{\alpha}{\varphi} \mid x : \varphi \in \Gamma \}$.
$\labs{\Gamma}$ denotes the set of all labels ocurring in the relational atoms and labelled formulas in $\Gamma$.

We interpret sequents with respect to PDL models using label \emph{valuations} $\valuation$, which are functions from labels to states.
We write $\model, \valuation \models x \relatedBy{a} y$ to mean that $(\valuation(x), \valuation(y)) \in \interp{a}$.
We write $\model, \valuation \models x : \varphi$ to mean $\model, \valuation(x) \models \varphi$.
For a sequent $\sequent{\Gamma}{\Delta}$, denoted by $S$, we write $\model, \valuation \models S$ to mean that $\model, \valuation \models B$ for \emph{some} $B \in \Delta$ whenever $\model, \valuation \models A$ for \emph{all} $A \in \Gamma$.
We write $\model, \valuation \not\models S$ whenever this is not the case, i.e.~when $\model, \valuation \models A$ for all $A \in \Gamma$ and $\model, \valuation \not\models B$ for all $B \in \Delta$.
We say $S$ is \emph{valid}, and write $\models S$, when $\model, \valuation \models S$ for all models $\model$ and valuations $\valuation$ that map each label to some state of $\model$.

\input{fig_pdlrules}

The sequent calculus {\InfinitarySystem} is defined by the inference rules in \cref{fig:PDLRules}.
A \emph{pre-proof} is a possibly infinite derivation tree built from these inference rules.

\begin{definition}[Pre-proof]
  A pre-proof is a possibly infinite (i.e.~non-well-found\-ed) derivation tree formed from inference rules. A path in a pre-proof is a possibly infinite sequence of sequents $s_0, s_1, \ldots (, s_n)$ such that $s_0$ is the root sequent of the proof, and $s_{i+1}$ is a premise of $s_i$
  for each $i < n$.
\end{definition}

Not all pre-proofs derive sound judgements.

\begin{example}
\label{eg:InvalidPreproof}
  The following pre-proof derives an invalid sequent.
  \begin{gather*}
    \begin{prooftree}
      \prooftree
        \prooftree
            \leadsto
          \sequent{}{x : \nec{\iter{a}}{p}}
        \endprooftree
          \using {\text{\smaller{(WR)}}}
          \justifies
        \sequent{}{x : \nec{\iter{a}}{p}, x : p}
      \endprooftree
      \prooftree
        \prooftree
            \leadsto
          \sequent{}{x : \nec{\iter{a}}{p}}
        \endprooftree
          \using {\text{\smaller{(WR)}}}
          \justifies
        \sequent{}{x : \nec{\iter{a}}{p}, x : \nec{a}{\nec{\iter{a}}{p}}}
      \endprooftree
        \using {\text{\smaller{($\ast$R)}}}
        \justifies
      \sequent{}{x : \nec{\iter{a}}{p}}
    \end{prooftree}
  \end{gather*}
  Note that, since our sequents consist of \emph{sets} of formulas, each instance of the ($\ast$R) rule incorporates a contraction
\end{example}

To distinguish pre-proofs deriving valid sequents, we define the notion of a trace through a pre-proof.
Traces consist of trace values, which (uniquely) identify particular modalities within labelled formulas.
$\vec{\alpha}_n$ denotes a sequence $\alpha_1, \ldots, \alpha_n$, and $\emptysequence$ denotes the empty sequence.
We sometimes omit the subscript indicating length, writing $\vec{\alpha}$, when irrelevant or evident from the context.

\begin{definition}[Trace Value]
  A \emph{trace value} $\tau$ is a tuple $(x, \vec{\alpha}_{n}, \beta, \varphi)$ consisting of a label $x$, a (possibly empty) sequence $\vec{\alpha}_{n}$ of $n$ programs, a program $\beta$, and a formula $\varphi$.
  We call $\vec{\alpha}$ the \emph{spine} of $\tau$, and $\beta$ the \emph{focus} of $\tau$.
  We write $\nec{\gamma}{\tau}$ for the trace value $(x, \gamma \cdot \vec{\alpha}_{n}, \beta, \varphi)$, and $y : \tau$ for the trace value $(y, \vec{\alpha}_{n}, \beta, \varphi)$.
  In an abuse of notation we also use $\tau$ to denote the corresponding labelled formula $x : \nec{\alpha_1}{\ldots \nec{\alpha_n}{\nec{\iter{\beta}}{\varphi}}}$.

\end{definition}

Trace values in the conclusion of an inference rule are related to trace values in its premises as follows.

\begin{definition}[Trace Pairs]
  Let $\tau$ and $\tau'$ be trace values, with sequents $\sequent{\Gamma}{\Delta}$ and $\sequent{\Gamma'}{\Delta'}$ (respectively denoted by $s$ and $s'$) the conclusion and a premise, respectively, of an inference rule $r$; we say that $(\tau, \tau')$ is a trace pair for $(s, s')$ when $\tau \in \Delta$ and $\tau' \in \Delta'$ and the following conditions hold.
  \begin{enumerate}[label={(\arabic*)}]
    \item
    If $\tau$ is the principal formula of the rule instance, then $\tau'$ is its immediate ancestor and moreover if the rule is an instance of:
    \begin{description}
      \item[{\normalfont{($\square$R)}}]
      then $\tau = x : \nec{a}{\tau'}$, where $x$ is the label of the principal formula;
      \item[{\normalfont{($?$R)}}]
      then $\tau = \nec{\test{\varphi}}{\tau'}$;
      \item[{\normalfont{($;$R)}}]
      then $\tau = \nec{\comp{\alpha}{\beta}}{\tau''}$ and $\tau' = \nec{\alpha}{\nec{\beta}{\tau''}}$ for some trace value $\tau''$;
      \item[{\normalfont{($\cup$R)}}]
      then there is some $\tau''$ such that:
        $\tau = \nec{\choice{\alpha}{\beta}}{\tau''}$;
        $\tau' = \nec{\alpha}{\tau''}$ if $s'$ is the left-hand premise; and
        $\tau' = \nec{\beta}{\tau''}$ if $s'$ is the right-hand premise;
      \item[{\normalfont{($\ast$R)}}]
      then $\tau = \nec{\iter{\alpha}}{\tau'}$ if $s'$ is the left-hand premise, and $\tau' = \nec{\alpha}{\tau}$ if $s'$ is the right-hand premise. 
    \end{description}
    \item
    If $\tau$ is not the principal formula of the rule then $\tau = x : \tau'$ if the rule is an instance of (Subst) and $x$ is the label substituted, and $\tau = \tau$ otherwise.
  \end{enumerate}
  If $\tau$ is the principal formula of the rule instance and the spine of $\tau$ is empty, then we say that the trace pair is \emph{progressing}.
\end{definition}

\noindent
Notice that when a trace pair is progressing for $(s, s')$, it is necessarily the case that the corresponding rule is an instance of ($\ast$R) and that $s'$ is the right-hand premise (although, not necessarily vice versa).

Traces along paths in a pre-proof consist of consecutive pairs of trace values for each corresponding step of the path.

\begin{definition}[Trace]
  \label{def:Trace}
  A \emph{trace} is a (possibly infinite) sequence of trace values.
  We say that a trace $\tau_1, \tau_2, \ldots (, \tau_n)$ \emph{follows} a path $s_1, s_2, \ldots (, s_m)$ in a pre-proof when there exists some $k \geq 0$ such that each consecutive pair of trace values $(\tau_{i}, \tau_{i+1})$ is a trace pair for $(s_{i+k}, s_{i+k+1})$; when $k = 0$, we say that the trace \emph{covers} the path.
  We say that the trace \emph{progresses} at $i$ if $(\tau_{i}, \tau_{i+1})$ is progressing, and say the trace is infinitely progressing if it progresses at infinitely many points.
\end{definition}

Proofs are pre-proofs that satisfy a well-formedness condition, called the global trace condition.

\begin{definition}[Infinite Proof]
  A {\InfinitarySystem} \emph{proof} is a pre-proof in which every infinite path is followed by some infinitely progressing trace.
\end{definition}

\begin{figure}[t!]
  \vspace{-2em}
  {\smaller\begin{gather*}
    \begin{prooftree}
      \prooftree
        \prooftree
          \prooftree
              \justifies
            \sequent{x : \nec{\iter{a}}{\varphi}}{x : \nec{\iter{a}}{\varphi}}
              \using {\text{\smaller{(Ax)}}}
          \endprooftree
          \prooftree
            \prooftree
              \prooftree
                \prooftree
                  \prooftree
                    \begin{gathered}
                      \tikz \coordinate (bud) at (0,0.5em) ;
                        \\[-1em]
                      \sequent{x : \nec{\iter{a}}{\varphi}}{x : \nec{\iter{{\color{blue}{\underline{a}}}}}{\nec{\iter{a}}{\varphi}}}
                    \end{gathered}
                      \justifies
                    \sequent{y : \nec{\iter{a}}{\varphi}}{y : \nec{\iter{{\color{blue}{\underline{a}}}}}{\nec{\iter{a}}{\varphi}}}
                      \using {\text{\smaller{(Subst)}}}
                  \endprooftree
                    \justifies
                  \sequent{x : \varphi, y : \nec{\iter{a}}{\varphi}}{y : \nec{\iter{{\color{blue}{\underline{a}}}}}{\nec{\iter{a}}{\varphi}}}
                    \using {\text{\smaller{(WL)}}}
                \endprooftree
                  \justifies
                \sequent{x \relatedBy{a} y, x : \varphi, x : \nec{a}{\nec{\iter{a}}{\varphi}}}{y : \nec{\iter{{\color{blue}{\underline{a}}}}}{\nec{\iter{a}}{\varphi}}}
                  \using {\text{\smaller{($\square$L)}}}
              \endprooftree
                \justifies
              \sequent{x : \varphi, x : \nec{a}{\nec{\iter{a}}{\varphi}}}{x : \nec{a}{\nec{\iter{{\color{blue}{\underline{a}}}}}{\nec{\iter{a}}{\varphi}}}}
                \using {\text{\smaller{($\square$R)}} \tikz \coordinate (right-edge) at (1em,0);}
            \endprooftree
              \justifies
            \sequent{x : \nec{\iter{a}}{\varphi}}{x : \nec{a}{\nec{\iter{{\color{blue}{\underline{a}}}}}{\nec{\iter{a}}{\varphi}}}}
              \using{\text{\smaller{($\ast$L)}}}
          \endprooftree
            \justifies
              \sequent{x : \nec{\iter{a}}{\varphi}}{x : \nec{{\color{blue}{\iter{\underline{a}}}}}{\nec{\iter{a}}{\varphi}}}
              \tikz \coordinate (companion) at (0.5em,0.25em) ;
            \using {\text{\smaller{($\ast$R)}}}
          \endprooftree
            \justifies
          \sequent{x : \nec{\iter{a}}{\varphi}}{x : \nec{\comp{\iter{a}}{\iter{a}}}{\varphi}}
            \using {\text{\smaller{($\compSymb$R)}}}
        \endprooftree
          \justifies
        \sequent{}{x : \nec{\iter{a}}{\varphi} \rightarrow \nec{\comp{\iter{a}}{\iter{a}}}{\varphi}}
          \using {\text{\smaller{($\rightarrow$R)}}}
    \end{prooftree}
    \begin{tikzpicture}[overlay]
      \draw[arrow,dashed] (bud) -- ++(0,1em) -| (right-edge) |- (companion) [->];
    \end{tikzpicture}
  \end{gather*}}
  \vspace{-1em}
  \caption{Representation of a {\InfinitarySystem} proof of $\nec{\iter{a}}{\varphi} \rightarrow \nec{\comp{\iter{a}}{\iter{a}}}{\varphi}$.}
  \label{fig:Example1}
\end{figure}

\begin{figure}[t!]
  \vspace{-1em}
  \scalebox{0.875}{\begin{minipage}{\textwidth}
  {\smaller\begin{gather*}
    \kern 1em
    \begin{prooftree}
      \prooftree
        \prooftree
          \prooftree
              \justifies
            \sequent{x : \varphi}{x : \varphi}
              \using {\text{\smaller{(Ax)}}}
          \endprooftree
            \justifies
          \sequent{x : \varphi, x : \nec{\iter{a}}{\nec{\iter{a}}{\varphi}}}{x : \varphi}
            \using {\text{\smaller{(WL)}}}
        \endprooftree
          \justifies
        \sequent{x : \nec{\iter{a}}{\varphi}}{x : \varphi}
          \using {\text{\smaller{($\ast$L)}}}
      \endprooftree
      \prooftree
        \begin{gathered}[b]
          \tikz \coordinate (bud1) at (0,-0.5em) ; \\
          \sequent{x : \nec{\iter{a}}{\varphi}}{x : \nec{\iter{\color{blue}{\underline{(\iter{a})}}}}{\varphi}}
        \end{gathered}
        \prooftree
          \prooftree
            \prooftree
              \prooftree
                \prooftree
                  \begin{gathered}[b]
                    \tikz \coordinate (bud2) at (0,-0.5em) ; \\
                    \sequent{x : \nec{\iter{a}}{\varphi}}{x : \nec{\iter{\color{red}{\overline{a}}}}{\nec{\iter{\color{blue}{\underline{(\iter{a})}}}}{\varphi}}}
                  \end{gathered}
                    \justifies
                  \sequent{y : \nec{\iter{a}}{\varphi}}{y : \nec{\iter{\color{red}{\overline{a}}}}{\nec{\iter{\color{blue}{\underline{(\iter{a})}}}}{\varphi}}}
                    \using {\text{\smaller{(Subst)}}}
                \endprooftree
                  \justifies
                \sequent{x : \varphi, y : \nec{\iter{a}}{\varphi}}{y : \nec{\iter{\color{red}{\overline{a}}}}{\nec{\iter{\color{blue}{\underline{(\iter{a})}}}}{\varphi}}}
                  \using {\text{\smaller{(WL)}}}
              \endprooftree
                \justifies
              \sequent{x \relatedBy{a} y, x : \varphi, x : \nec{a}{\nec{\iter{a}}{\varphi}}}{y : \nec{\iter{\color{red}{\overline{a}}}}{\nec{\iter{\color{blue}{\underline{(\iter{a})}}}}{\varphi}}}
                \using {\text{\smaller{($\square$L)}}}
            \endprooftree
              \justifies
            \sequent{x : \varphi, x : \nec{a}{\nec{\iter{a}}{\varphi}}}{x : \nec{a}{\nec{\iter{\color{red}{\overline{a}}}}{\nec{\iter{\color{blue}{\underline{(\iter{a})}}}}{\varphi}}}}
              \using {\text{\smaller{($\square$R)}}}
          \endprooftree
            \justifies
          \sequent{x : \nec{\iter{a}}{\varphi}}{x : \nec{a}{\nec{\iter{\color{red}{\overline{a}}}}{\nec{\iter{\color{blue}{\underline{(\iter{a})}}}}{\varphi}}}}
            \using {\text{\smaller{($\ast$L)}}}
        \endprooftree
          \justifies
        \sequent{x : \nec{\iter{a}}{\varphi}}{x : \nec{\iter{\color{red}{\overline{a}}}}{\nec{\iter{\color{blue}{\underline{(\iter{a})}}}}{\varphi}}}
        \tikz \coordinate (companion2) at (0.25em,0.25em) ;
          \using {\text{\smaller{($\ast$R)}} \; \ddag}
      \endprooftree
        \justifies
      \tikz \coordinate (companion1) at (-0.25em,0.25em) ;
      \sequent{x : \nec{\iter{a}}{\varphi}}{x : \nec{\iter{\color{blue}{\underline{(\iter{a})}}}}{\varphi}}
        \using {\text{\smaller{($\ast$R)}} \; \dag}
    \end{prooftree}
    \begin{tikzpicture}[overlay]
      \draw[arrow,dashed] (bud1) -- ++(0,4em) -- ++(-21em,0) |- (companion1) [->];
      \draw[arrow,dashed] (bud2) -- ++(0,1em) -- ++(12.5em,0) |- (companion2) [->];
    \end{tikzpicture}
  \end{gather*}}
\end{minipage}}
  \vspace{-1em}
  \caption{Representation of a {\InfinitarySystem} proof of $\sequent{x : \nec{\iter{a}}{\varphi}}{x : \nec{\iter{(\iter{a})}}{\varphi}}$.}
  \label{fig:Example2}
\end{figure}

\begin{example}
  \Cref{fig:Example1} shows a finite representation of a {\InfinitarySystem} proof of the formula $\nec{\iter{a}}{\varphi} \rightarrow \nec{\comp{\iter{a}}{\iter{a}}}{\varphi}$.
  The full infinite proof can be obtain by unfolding the cycle an infinite number of times.
  An infinitely progressing trace following the (unique) infinite path in this proof is indicated by the underlined programs highlighted in blue, which denote the focus of the trace value in each sequent.
  The progression point is the (only) instance of the ($\ast$R) rule.

  \Cref{fig:Example2} shows a finite representation of a {\InfinitarySystem} proof of the sequent $\sequent{x : \nec{\iter{a}}{\varphi}}{x : \nec{\iter{(\iter{a})}}{\varphi}}$.
  This proof is more complex than that of \cref{fig:Example1}, and involves two overlapping cycles.
  This proof contains more than one infinite path (in fact, it contains an infinite number of infinite paths).
  However, they fall into three categories: 
  \begin{enumerate*}[label={(\arabic*)}]
    \item
    those that eventually traverse only the upper cycle;
    \item
    those that eventually traverse only the lower cycle; and
    \item
    those that traverse both cycles infinitely often.
  \end{enumerate*}
  Infinite paths of the first variety have an infinitely progressing trace indicated by the overlined programs highlighted in red.
  The progression point is the upper instance of ($\ast$R) rule, marked by $(\ddag)$.
  The remaining infinite paths have a trace 
  indicated by the underlined programs highlighted in blue.
  This trace does not progress around the upper cycle (for those paths that traverse it), but does progress once around each lower cycle at the instance of the ($\ast$R) rule marked by $(\dag)$.
  Since these paths traverse this lower cycle infinitely often, the trace is infinitely progressing.
\end{example}

\begin{remark}
  The notion of trace in the system for Kleene Algebra of Das and Pous \cite{DasPous17,DasPous18} appears simpler than ours: a sequence of formulas (on the left) connected by ancestry, with such a trace being valid if it is principal for a (left) unfolding rule infinitely often.
  In fact, we can show that our definition of trace is equivalent to an analogous formulation of this notion for our system.
  However, our definition allows for a direct, semantic proof of soundness via infinite descent.
  In contrast, the soundness proof in \cite{DasPous18} relies on cut-admissibility and an inductive proof-theoretic argument for the soundness of the cut-free fragment.
  It is unclear that a similar technique can be used to show soundness of the cut-free fragment of our system.  
  Furthermore, the cut-free fragment of the system of Das and Pous is notable in that it admits a simpler trace condition than the full system: namely, that every infinite path is fair for the (left) unfolding rule \cite[prop.~8]{DasPous18}.
  Our system does not satisfy this property, due to the ability to perform contraction and weakening, as demonstrated in \cref{eg:InvalidPreproof}.
\end{remark}

The proof system is sound since, for invalid sequents, we can map traces to decreasing sets of counter-examples in (finite) models.

A path in a model $\model$ is a sequence of states $s_1, \ldots, s_n$ in $\model$ such that each successive pair of states satisfies $(s_i, s_{i+1}) \in \interp[\model]{a}$ for some $a$.
A path in $\model$ is called \emph{loop-free} if it does not contain any repeated states.
If $\vec{s}$ and $\vec{s'}$ are paths in $\model$, we write $\vec{s} \prefix \vec{s'}$ to denote that $\vec{s}$ is a prefix of $\vec{s'}$.

An $m$-partition of a path $\vec{s}_n$ is a sequence of $m$ increasing indices $k_1 \leq \ldots \leq k_m \leq n$.
A path in $\model$ for a trace value $\tau = (x, \vec{\alpha}_{n}, \beta, \varphi)$ with respect to a valuation $\valuation$ is a path $\vec{s}_m$ in $\model$ with $s_1 = \valuation(x)$ having an $n$-partition $k_1, \ldots, k_{n}$ satisfying $(s_{k_i}, s_{k+1}) \in \interp[\model]{\alpha_{i+1}}$ for each $0 \leq i < n$ and $(s_{k_n}, s_m) \in \interp[\model]{\iter{\beta}}$, where we take $k_0 = 1$ (i.e.~$s_{k_0} = s_1$).
The $n$-partition $k_1, \ldots, k_{n}$ is called a partition of $\vec{s}_m$ for $\tau$.
A counter-example in $\model$ for a trace value $\tau$ at $\valuation$ is simply a path $\vec{s}_m$ in $\model$ for $\tau$ w.r.t.~$\valuation$ such that $\model, s_m \not\models \varphi$.

A given path in $\model$ for $\tau$ at $\valuation$ can, in general, have many different partitions.
A partition $\vec{k}_n$ of a path $\vec{s}_m$ for $\tau$ at $\valuation$ is called \emph{maximal} if the length of its final segment $s_{k_n}, \ldots, s_m$ is maximal among all such partitions.
We define the \emph{weight} of a path $\vec{s}$ in $\model$ for $\tau$ at $\valuation$ to be the length of the final segments of its maximal partition(s).
We denote this by $\sizeOf[(\model, \valuation)]{\vec{s}}{\tau}$.
If $\Pi$ is a set of paths in $\model$ for $\tau$ at $\valuation$, we define the \emph{measure} of $\Pi$, denoted $\pathsetMeasure{\Pi}$, to be the multiset of weights of the paths it contains; that is $\pathsetMeasure{\Pi} = \{ \sizeOf[(\model, \valuation)]{\vec{s}}{\tau} \mid \vec{s} \in \Pi \}$.

The measure for trace values in a model $\model$ at a valuation $\valuation$, then, is simply the measure of the set of all of its `nearest' counter-examples.

\begin{definition}[Trace Value Measure]
\label{def:TraceValueMeasure}
  Let $\counterEgs{\model}{\valuation}{\tau}$ denote the set of all loop-free counter-examples $\vec{s}$ in $\model$ for $\tau$ at $\valuation$ such that there is no counter-example $\vec{s'}$ in $\model$ for $\tau$ at $\valuation$ with $\vec{s'} \prefix \vec{s}$.
  The measure of $\tau$ in $\model$ at $\valuation$ is defined as $\measureOf{\model}{\valuation}{\tau} = \pathsetMeasure{\counterEgs{\model}{\valuation}{\tau}}$.
\end{definition}

For finitely branching models $\model$, it is clear that trace value measures are always finite.
Note that finite multisets $M$ of elements of a well-ordering can be well-ordered using, e.g., the Dershowitz-Manna ordering $<_{\text{DM}}$ \cite{DershowitzManna79}.
This means that we have the following property.

\begin{restatable}[Descending Counter-models]{lemma}{DescendingCountermodels}
\label{lem:DescendingCountermodels}
  Let $\sequent{\Gamma}{\Delta}$, denoted $S$, be the conclusion of an instance of an inference rule, and suppose there is a finitely branching model $\model$ and  valuation $\valuation$ such that $\model, \valuation \not\models S$, then there is a premise $\sequent{\Gamma'}{\Delta'}$ of the rule instance, denoted $S'$, and a valuation $\valuation'$ such that $\model, \valuation' \not\models S'$ and for each trace pair $(\tau, \tau')$ for $(S, S')$, $\measureOf{\model}{\valuation'}{\tau'} \leq_{\text{DM}} \measureOf{\model}{\valuation}{\tau}$ and also $\measureOf{\model}{\valuation'}{\tau'} <_{\text{DM}} \measureOf{\model}{\valuation}{\tau}$ if $(\tau, \tau')$ is progressing.
\end{restatable}

This entails the soundness of our proof system, since PDL has the finite model property \cite[Thm.~3.2]{FischerLadner1979}.
This property states that, if a PDL formula is satisfiable, then it is satisfiable in a finite (and thus finitely branching) model.
Thus, if a sequent is not valid then there is a finitely branching model that falsifies it.
If a {\InfinitarySystem} proof $\mathcal{P}$ were to derive an invalid sequent, then by \cref{lem:DescendingCountermodels} it would contain an infinite path $\sequent{\Gamma_1}{\Delta_1}, \sequent{\Gamma_2}{\Delta_2}, \ldots$ for which there exists a finite model $\model$ and a matching sequence of valuations $\valuation_1, \valuation_2, \ldots$ that invalidate each sequent in the path.
Moreover, these invalidating valuations ensure that the measures of the trace values in any trace pair along the path is decreasing, and strictly so for progressing trace pairs.
However, since $\mathcal{P}$ is a proof, it satisfies the global trace condition.
This means that there would be an infinitely progressing trace following the path $\sequent{\Gamma_1}{\Delta_1}, \sequent{\Gamma_2}{\Delta_2}, \ldots$ and thus we would be able to construct an infinitely descending chain of (finite) trace value measures.
Because the set of finite trace value measures is well-founded, this is impossible and so the derived sequent must in fact be valid.

\begin{restatable}[Soundness]{theorem}{Soundness}
  \label{thm:Soundness:Inf}
  {\InfinitarySystem} derives only valid sequents.
\end{restatable}



The \emph{cyclic} system {\CyclicSystem} is obtained by restricting consideration to only those proofs of {\InfinitarySystem} that are \emph{regular}, i.e.~have only a finite number of distinct subtrees.

\begin{definition}[Cyclic Pre-proof]
  A \emph{cyclic} pre-proof is a pair $(P, f)$ consisting of a finite derivation tree $P$ possibly containing open leaves 
  called \emph{buds}, and a function $f$ assigning to each bud an internal node of the tree, called its \emph{companion}, with a syntactically identical sequent.
\end{definition}

We usually represent a cyclic pre-proof as the graph induced by identifying each bud with its companion (as in \cref{fig:Example1,fig:Example2}).
The infinite unfolding of a cyclic pre-proof is the {\InfinitarySystem} pre-proof obtained as the limit of the operation that replaces each bud with a copy of the subderivation concluding with its companion an infinite number of times.
A cyclic proof is a cyclic pre-proof whose infinite unfolding satisfies the global trace condition.
As in other cyclic systems (e.g.~\cite{BrotherstonSimpson11,CohenRowe18,RoweBrotherston17,TellezBrotherston17}) it is decidable whether or not this is the case via a construction involving complementation of B\"{u}chi automata.
This means that decidability of the global trace condition for {\CyclicSystem} pre-proofs is \textsf{PSPACE}-complete.

Since every {\CyclicSystem} is also a {\InfinitarySystem} proof, soundness of the cyclic system is an immediate corollary of \cref{thm:Soundness:Inf}.

\begin{corollary}
  \label{cor:Soundness:Cyclic}
  If $\sequent{\Gamma}{\Delta}$ is derivable in {\CyclicSystem}, then $\sequent{\Gamma}{\Delta}$ is valid.
\end{corollary}

\section{Completeness}

In this section, we give completeness results for our systems.
We show that the full system, {\InfinitarySystem}, is cut-free complete.
On the other hand, if we allow instances of the (Cut) rule, then every valid theorem of PDL has a proof in the cyclic subsystem {\CyclicSystem}.

\subsection{Cut-free Completeness of {\InfinitarySystem}}

We use a standard technique of defining a pre-proof that encodes an exhaustive search for a cut-free proof (as used in, e.g., \cite{BrotherstonSimpson11,CohenRowe18}).
For invalid sequents, this results in a pre-proof from which we can construct a counter-model, using the formulas that occur along a particular path.

A \emph{schedule} $\sigma$ is an enumeration of labelled non-atomic formulas in which each labelled formula occurs infinitely often.
The $i$\textsuperscript{th} element of $\sigma$ is written $\sigma_i$.

\begin{definition}[Search Tree]
\label{def:SearchTree}
  Given a sequent $\sequent{\Gamma}{\Delta}$ and a schedule $\sigma$, we can define an infinite sequence $\vec{\mathcal{D}}$ of open derivations inductively.
  Taking $\mathcal{D}_{0} = \sequent{\Gamma}{\Delta}$, we construct each $\mathcal{D}_{i+1}$ from its predecessor $\mathcal{D}_{i}$ by:
  \begin{enumerate}[wide,noitemsep]
    \item
    firstly closing any open leaves $\sequent{\Gamma'}{\Delta'}$ for which $x : \falsum \in \Gamma$ for some $x$ or $\Gamma \cap \Delta \neq \emptyset$ by applying weakening rules leading to an instance of {\normalfont{($\falsum$)}} or an axiom $\sequent{A}{A}$ for some $A \in \Gamma \cap \Delta$ (thus the antecedent of each remaining open node is disjoint from its consequent);
    \item
    then replacing each remaining open node $\sequent{\Gamma'}{\Delta'}$ in which $\sigma_{i}$ occurs with applications of the rule for which $\sigma_i$ is principal in the following way.
    \begin{itemize}[nosep]
      \item
      If $\sigma_i = x : \nec{a}{\varphi} \in \Delta'$, then we pick a label $y$ not ocurring in $\sequent{\Gamma'}{\Delta'}$, and replace the open node with the following derivation.
      \begin{gather*}
        \begin{prooftree}
          \sequent{x \relatedBy{a} y, \Gamma'}{\Delta', y : \varphi}
            \using {\text{\normalfont{\smaller{($\square$R)}}}}
            \justifies
          \sequent{\Gamma'}{\Delta', x : \nec{a}{\varphi}}
        \end{prooftree}
      \end{gather*}
      \item
      If $\sigma_i = x : \nec{a}{\varphi} \in \Gamma'$ then, letting $\{ y_1, \ldots, y_n \}$ be the set of all $y_i$ such that $x \relatedBy{a} y_i \in \Gamma'$, we replace the open node with the following derivation.
      \begin{gather*}
        \begin{prooftree}
          \prooftree
            \prooftree
              \sequent{x : \nec{a}{\varphi}, \{ y_1 : \varphi, \ldots, y_n : \varphi \}, \Gamma'}{\Delta'}
                \leadsto
              \sequent{x : \nec{a}{\varphi}, \{ y_1 : \varphi, y_2 : \varphi \}, \Gamma'}{\Delta'}
            \endprooftree
              \using {\text{\normalfont{\smaller{($\square$L)}}}}
              \justifies
            \sequent{x : \nec{a}{\varphi}, \{ y_1 : \varphi \}, \Gamma'}{\Delta'}
          \endprooftree
            \using {\text{\normalfont{\smaller{($\square$L)}}}}
            \justifies
          \sequent{x : \nec{a}{\varphi}, \Gamma'}{\Delta'}
        \end{prooftree}
      \end{gather*}
      \item
      In all other cases, we replace the open node with an application of the appropriate rule {\normalfont{($r$)}} as follows, where $\Gamma'_i$ and $\Delta'_i$, $i \in \{1, 2\}$, are the sets of left and right immediate ancestors of $\sigma_i$, respectively, for the appropriate premise.
      \begin{gather*}
        \begin{prooftree}
          \sequent{\Gamma'_1, \Gamma'}{\Delta', \Delta'_1}
            \quad
          (\sequent{\Gamma'_2, \Gamma'}{\Delta', \Delta'_2})
            \using {\text{\normalfont{\smaller{($r$)}}}}
            \justifies          
          \sequent{\Gamma'}{\Delta'}
        \end{prooftree}
      \end{gather*}
    \end{itemize}
  \end{enumerate}
  Since each $\mathcal{D}_{i}$ is a prefix of $\mathcal{D}_{i+1}$, there is a smallest derivation containing each $\mathcal{D}_{i}$ as a prefix.
  We call this derivation a \emph{search tree} for $\sequent{\Gamma}{\Delta}$ (w.r.t.~$\sigma$).
\end{definition}

Notice that search trees do not contain instances of the (Cut) or (Subst) rules.
Moreover, when a given search tree $\mathcal{D}$ is not a valid proof, we may extract from it two sets of labelled formulas and relational atoms that we can use to construct a countermodel.
If $\mathcal{D}$ is not a valid proof, then either it contains an open node to which no schedule element applies or it contains an infinite path that does not satisfy the global trace condition (an \emph{untraceable} branch).
For a search tree $\mathcal{D}$, we say that a pair $(\Gamma, \Delta)$ is a \emph{template induced by $\mathcal{D}$} when either:
\begin{enumerate*}[label={(\roman*)}]
  \item
  $\sequent{\Gamma}{\Delta}$ is an open node of $\mathcal{D}$; or
  \item
  $\Gamma = \bigcup_{i > 0} \Gamma_i$ and $\Delta = \bigcup_{i > 0} \Delta_i$, where $\sequent{\Gamma_1}{\Delta_1}, \sequent{\Gamma_2}{\Delta_2}, \ldots$ is an untraceable branch in $\mathcal{D}$.
\end{enumerate*}
Notice that, due to the construction of search trees, the component sets of a template are necessarily disjoint.
Given a template, we construct a PDL model as follows.

\begin{definition}[Countermodel Construction]
\label{def:CanonicalModel}
  Let $P = (\Gamma, \Delta)$ be a template induced by a search tree.
  The PDL model determined by the template P is given by $\model_{P} = (\Labels, \mathcal{I}_{P})$, where $\mathcal{I}_{P}$ is the following interpretation function:
  \begin{enumerate}
    \item
    $\mathcal{I}_{P}(p) = \{ x \mid x : p \in \Gamma \}$ for each atomic proposition $p$; and
    \item
    $\mathcal{I}_{P}(a) = \{ (x, y) \mid x \relatedBy{a} y \in \Gamma \}$ for each atomic program $a$.
  \end{enumerate}
  We write $\idValuation$ for the valuation defined by $\idValuation(x) = x$ for each label $x$.
\end{definition}

PDL models determined by templates have the following property.

\begin{restatable}{lemma}{CanonicalModels}
\label{lem:CanonicalModel}
  Let $P = (\Gamma, \Delta)$ be a template induced by a search tree.
  Then we have $\model_{P}, \idValuation \models A$ for all $A \in \Gamma$ and $\model_{P}, \idValuation \not\models B$ for all $B \in \Delta$.
\end{restatable}

\Cref{lem:CanonicalModel} entails the cut-free completeness of {\InfinitarySystem}.

\begin{restatable}[Completeness of {\InfinitarySystem}]{theorem}{CutfreeCompleteness}
\label{thm:Cut-freeCompleteness}
  If $\sequent{\Gamma}{\Delta}$ is valid, then it has a cut-free {\InfinitarySystem} proof.
\end{restatable}

\subsection{Completeness of {\CyclicSystem} for PDL}

We show that the cyclic system {\CyclicSystem} can derive all theorems of PDL by demonstrating that it can derive each of the axiom schemas and inference rules in \cref{fig:PDLAxiomatisation}, which (along with the axiom schemas of classical propositional logic) constitute a complete axiomatisation of PDL \cite[\textsection{7.1}]{HKT2000}.

\begin{figure}[t]
{%
\setlength{\abovedisplayskip}{-1em}%
\setlength{\belowdisplayskip}{-1em}%
\setlength{\abovedisplayshortskip}{-1em}%
\setlength{\belowdisplayshortskip}{-1em}%
\begin{tabularx}{\textwidth}{@{}XX@{}}
  \begin{equation}
    \parbox[c][1em]{0.395\textwidth}{\centering$
      \nec{\alpha}{(\varphi \rightarrow \psi)} \rightarrow (\nec{\alpha}{\varphi} \rightarrow \nec{\alpha}{\psi})
    $}
    \label[axiom]{ax:PDL:Distribution:Implication}
  \end{equation}
    &
  \begin{equation}
    \parbox[c][1em]{0.395\textwidth}{\centering$
      \nec{\alpha}{(\varphi \wedge \psi)} \rightarrow (\nec{\alpha}{\varphi} \wedge \nec{\alpha}{\psi})
    $}
    \label[axiom]{ax:PDL:Distribution:Conjunction}
  \end{equation}
    \\
  \begin{equation}
    \parbox[c][1em]{0.395\textwidth}{\centering$
      \nec{\choice{\alpha}{\beta}}{\varphi} \leftrightarrow \nec{\alpha}{\varphi} \wedge \nec{\beta}{\varphi}
    $}
    \label[axiom]{ax:PDL:Choice}
  \end{equation}
    &
  \begin{equation}
    \parbox[c][1em]{0.395\textwidth}{\centering$
      \nec{\comp{\alpha}{\beta}}{\varphi} \leftrightarrow \nec{\alpha}{\nec{\beta}{\varphi}}
    $}
    \label[axiom]{ax:PDL:Sequence}
  \end{equation}
    \\
  \begin{equation}
    \parbox[c][1em]{0.395\textwidth}{\centering$
      \nec{\test{\psi}}{\varphi} \leftrightarrow (\psi \rightarrow \varphi)
    $}
    \label[axiom]{ax:PDL:Test}
  \end{equation}
    &
  \begin{equation}
    \parbox[c][1em]{0.395\textwidth}{\centering$
      \varphi \wedge \nec{\alpha}{\nec{\iter{\alpha}}{\varphi}} \leftrightarrow \nec{\iter{\alpha}}{\varphi}
    $}
    \label[axiom]{ax:PDL:Mix}
  \end{equation}
    \\
  \begin{equation}
    \parbox[c][1em]{0.395\textwidth}{\centering$
      \varphi \wedge \nec{\iter{\alpha}}{(\varphi \rightarrow \nec{\alpha}{\varphi})} \rightarrow \nec{\iter{\alpha}}{\varphi}
    $}
    \label[axiom]{ax:PDL:Induction}
  \end{equation}
    \\[0.5em]
  \begin{equation}
    \begin{prooftree}
      \varphi
        \quad
      \varphi \rightarrow \psi
        \justifies
      \psi
    \end{prooftree}
    \tag{\textsf{MP}}
    \label[axiom]{ax:ModusPonens}
  \end{equation}
    &
  \begin{equation}
    \begin{prooftree}
      \varphi
        \justifies
      \nec{\alpha}{\varphi}
    \end{prooftree}
    \tag{\textsf{Nec}}
    \label[axiom]{ax:PDL:Necessitation}
  \end{equation}
\end{tabularx}}
  \caption{Axiomatisation of PDL.}
  \label{fig:PDLAxiomatisation}
\end{figure}

The derivation of the axioms of classical propositional logic is standard, and \crefrange{ax:PDL:Choice}{ax:PDL:Mix} are immediately derivable via the left and right proof rules for their corresponding syntactic constructors.
Each such derivation is finite, and thus trivially a {\CyclicSystem} proof.
\Cref{ax:PDL:Distribution:Implication,ax:PDL:Distribution:Conjunction,ax:PDL:Induction,ax:PDL:Necessitation} require the following lemma showing that a general form of necessitation is derivable.

\begin{restatable}[Necessitation]{lemma}{Necessitation}
  \label{lem:Necessitation}
  For any labelled formula $x : \varphi$, program $\alpha$, and finite set $\Gamma$ of labelled formulas such that $\labs{\Gamma} = \{ x \}$, there exists a $\CyclicSystem$ derivation concluding with the sequent $\sequent{\nec{\alpha}{\Gamma}}{x : \nec{\alpha}{\varphi}}$ and containing open leaves of the form $\sequent{\Gamma}{x : \varphi}$ such that:
  \begin{conditionenum}[nosep,wide,label={(\roman*)},ref={\roman*}]
    \item \label{cond:paths:finite}
    for each trace value $\tau = x : \varphi$, every path from the conclusion to an open leaf is covered by a trace $\nec{\alpha}{\tau}, \ldots, \tau$; and
    \item \label{cond:paths:infinite}
    every infinite path is followed by an infinitely progressing trace.
  \end{conditionenum}
\end{restatable}

\begin{figure}[t]
  \subfloat[Derivation schema for \Cref{ax:PDL:Distribution:Implication}]{
    \label{fig:Distribution:Implication}
    \scalebox{0.85}{\begin{minipage}[b]{0.5\textwidth}
      \begin{gather*}
        \begin{prooftree}
          \prooftree
            \prooftree
              \prooftree
                \prooftree
                    \justifies
                  \sequent{x : \varphi}{x : \varphi}
                    \using {\text{\smaller{(Ax)}}}
                \endprooftree
                \prooftree
                    \justifies
                  \sequent{x : \psi}{x : \psi}
                    \using {\text{\smaller{(Ax)}}}
                \endprooftree
                  \justifies
                \sequent{x : \varphi \rightarrow \psi, x : \varphi}{x : \psi}
                  \using {\text{\smaller{($\rightarrow$L)}}}
              \endprooftree
                \leadsto\proofdotnumber=6
              {\sequent{x : \nec{\alpha}{\varphi \rightarrow \psi}, x : \nec{\alpha}{\varphi}}{x : \nec{\alpha}{\psi}}}
                \using {\hbox to -2pt {\text{\smaller{\cref{lem:Necessitation}}}}}
            \endprooftree
              \justifies
            \sequent{x : \nec{\alpha}{\varphi \rightarrow \psi}}{x : \nec{\alpha}{\varphi} \rightarrow \nec{\alpha}{\psi}}
              \using {\text{\smaller{($\rightarrow$R)}}}
          \endprooftree
            \justifies
          \sequent{}{x : \nec{\alpha}{(\varphi \rightarrow \psi)} \rightarrow (\nec{\alpha}{\varphi} \rightarrow \nec{\alpha}{\psi})}
            \using {\text{\smaller{($\rightarrow$R)}}}
        \end{prooftree}
      \end{gather*}
      \vspace{-0.5em}
    \end{minipage}}}
    \hspace{1em}
    \subfloat[Derivation schema for \Cref{ax:PDL:Distribution:Conjunction}]{
      \label{fig:Distribution:Conjunction}
      \scalebox{0.76}{\begin{minipage}[b]{0.7\textwidth}
        \begin{gather*}
          \begin{prooftree}
            \prooftree
              \prooftree
                \prooftree
                  \prooftree
                    \prooftree
                        \justifies
                      \sequent{x : \varphi}{x : \varphi}
                        \using {\text{\smaller{(Ax)}}}
                    \endprooftree
                      \justifies
                    \sequent{x : \varphi, x : \psi}{x : \varphi}
                      \using {\text{\smaller{(WL)}}}
                  \endprooftree
                    \justifies
                  \sequent{x : \varphi \wedge \psi}{x : \varphi}
                    \using {\text{\smaller{($\wedge$L)}}}
                \endprooftree
                  \leadsto\proofdotnumber=6
                \sequent{x : \nec{\alpha}{(\varphi \wedge \psi)}}{x : \nec{\alpha}{\varphi}}
                  \using {\hbox to -2pt {\text{\smaller{\cref{lem:Necessitation}}}}}
              \endprooftree
              \prooftree
                \prooftree
                  \prooftree
                    \prooftree
                        \justifies
                      \sequent{x : \psi}{x : \psi}
                        \using {\text{\smaller{(Ax)}}}
                    \endprooftree
                      \justifies
                    \sequent{x : \varphi, x : \psi}{x : \psi}
                      \using {\text{\smaller{(WL)}}}
                  \endprooftree
                    \justifies
                  \sequent{x : \varphi \wedge \psi}{x : \psi}
                    \using {\text{\smaller{($\wedge$L)}}}
                \endprooftree
                  \leadsto\proofdotnumber=6
                \sequent{x : \nec{\alpha}{(\varphi \wedge \psi)}}{x : \nec{\alpha}{\varphi}}
                  \using {\hbox to -2pt {\text{\smaller{\cref{lem:Necessitation}}}}}
              \endprooftree
                \justifies
              \sequent{x : \nec{\alpha}{(\varphi \wedge \psi)}}{x : \nec{\alpha}{\varphi} \wedge \nec{\alpha}{\psi}}
                \using {\text{\smaller{($\wedge$R)}}}
            \endprooftree
              \justifies
            \sequent{}{x : \nec{\alpha}{(\varphi \wedge \psi)} \rightarrow (\nec{\alpha}{\varphi} \wedge \nec{\alpha}{\psi})}
              \using {\text{\smaller{($\rightarrow$R)}}}
          \end{prooftree}
        \end{gather*}
        \vspace{-0.5em}
      \end{minipage}}}
        \\[-1em]
    \subfloat[Derivation schema for \Cref{ax:PDL:Induction}]{
      \label{fig:Induction}
      \scalebox{0.75}{\begin{minipage}[b]{1.25\textwidth}
        \begin{gather*}
          \begin{prooftree}
            \prooftree
              \prooftree
                \prooftree
                  \prooftree
                      \justifies
                    \sequent{x : \varphi}{x : \varphi}
                      \using {\text{\smaller{(Ax)}}}
                  \endprooftree
                    \justifies
                  \sequent{x : \varphi, x : \nec{\iter{\alpha}}{\varphi \rightarrow \nec{\alpha}{\varphi}}}{x : \varphi}
                    \using {\text{\smaller{(WL)}}}
                \endprooftree
                \prooftree
                  \prooftree
                    \prooftree
                        \justifies
                      \sequent{x : \varphi}{x : \varphi}
                        \using {\text{\smaller{(Ax)}}}
                    \endprooftree
                    \prooftree
                      \prooftree
                        \begin{gathered}[b]
                          \tikz \coordinate (bud) at (0em,-0.5em) ;
                            \\
                            \sequent{x : \varphi, x : \nec{\iter{\alpha}}{\varphi \rightarrow \nec{\alpha}{\varphi}}}{x : \nec{\iter{{\color{blue}{\underline{\alpha}}}}}{\varphi}}%
                        \end{gathered}
                          \leadsto\proofdotnumber=6
                        \sequent{x : \nec{\alpha}{\varphi}, x : \nec{\alpha}{\nec{\iter{\alpha}}{\varphi \rightarrow \nec{\alpha}{\varphi}}}}{x : \nec{\alpha}{\nec{\iter{{\color{blue}{\underline{\alpha}}}}}{\varphi}}}
                          \using {\hbox to -2pt {\text{\smaller{\cref{lem:Necessitation}}}}}
                      \endprooftree
                        \justifies
                      \sequent{x : \varphi, x : \nec{\alpha}{\varphi}, x : \nec{\alpha}{\nec{\iter{\alpha}}{\varphi \rightarrow \nec{\alpha}{\varphi}}}}{x : \nec{\alpha}{\nec{\iter{{\color{blue}{\underline{\alpha}}}}}{\varphi}}}
                        \using {\text{\smaller{(WL)}}}
                    \endprooftree
                      \justifies
                    \sequent{x : \varphi, x : \varphi \rightarrow \nec{\alpha}{\varphi}, x : \nec{\alpha}{\nec{\iter{\alpha}}{\varphi \rightarrow \nec{\alpha}{\varphi}}}}{x : \nec{\alpha}{\nec{\iter{{\color{blue}{\underline{\alpha}}}}}{\varphi}}}
                      \using {\text{\smaller{($\rightarrow$L)}}}
                  \endprooftree
                    \justifies
                  \sequent{x : \varphi, x : \nec{\iter{\alpha}}{\varphi \rightarrow \nec{\alpha}{\varphi}}}{x : \nec{\alpha}{\nec{\iter{{\color{blue}{\underline{\alpha}}}}}{\varphi}}}
                    \using {\text{\smaller{($\ast$L)}}}
                \endprooftree
                  \justifies
                {\sequent{x : \varphi, x : \nec{\iter{\alpha}}{\varphi \rightarrow \nec{\alpha}{\varphi}}}{x : \nec{\iter{{\color{blue}{\underline{\alpha}}}}}{\varphi}}
                \tikz \coordinate (companion) at (0.5em,0.25em) ; }
                  \using {\text{\smaller{($\ast$R)}}}
              \endprooftree
                \justifies
              \sequent{x : \varphi \wedge \nec{\iter{\alpha}}{\varphi \rightarrow \nec{\alpha}{\varphi}}}{x : \nec{\iter{\alpha}}{\varphi}}
                \using {\text{\smaller{($\wedge$L)}}}
            \endprooftree
              \justifies
            \sequent{}{x : \varphi \wedge \nec{\iter{\alpha}}{\varphi \rightarrow \nec{\alpha}{\varphi}} \rightarrow \nec{\iter{\alpha}}{\varphi}}
              \using {\text{\smaller{($\rightarrow$R)}}}
          \end{prooftree}
          \begin{tikzpicture}[overlay]
            \draw[arrow,dashed] (bud) -- ++(0,1em) -- ++(13.5em,0) |- (companion) [->];
          \end{tikzpicture}
        \end{gather*}
        \vspace{-0.5em}
      \end{minipage}}}
  \caption{$\CyclicSystem$ derivation schemata for the distribution and induction axioms.}
  \label{fig:DistributionInduction}
\end{figure}

Schemas for deriving \cref{ax:PDL:Distribution:Implication,ax:PDL:Distribution:Conjunction,ax:PDL:Induction} are shown in \cref{fig:DistributionInduction}.
Any infinite paths which exist in the schemas for deriving \cref{ax:PDL:Distribution:Implication,ax:PDL:Distribution:Conjunction} are followed by infinitely progressing traces by \cref{lem:Necessitation}.
Thus, they are $\CyclicSystem$ proofs.
In the schema for \cref{ax:PDL:Induction}, the open leaves of the subderivation constructed via \cref{lem:Necessitation} are converted into buds, the companion of each of which is the conclusion of the instance of the ($\ast$R) rule.
\Cref{cond:paths:finite} of \cref{lem:Necessitation} guarantees that each infinite path along these cycles has an infinitely progressing trace.
We thus have the following completeness result.

\begin{restatable}{theorem}{RegularCompleteness}
\label{thm:Completeness:Regular}
  If $\varphi$ is valid then $\sequent{}{x : \varphi}$ is derivable in {\CyclicSystem}.
\end{restatable}

It should be noted that \cref{thm:Completeness:Regular} is not a deductive completeness result, i.e.~it does not say that any sequent $\sequent{\Gamma}{\Delta}$ is only valid if there is a {\CyclicSystem} proof for it. This is no major restriction, as a finitary syntactic consequence relation 
cannot capture semantic consequence in PDL: due to the presence of iteration, PDL is not compact.
This can only be rectified by allowing infinite sequents in the proof system, which is undesirable for present purposes.

\section{Proof Search for Test-free, Acyclic Sequents}

In this section, we describe a cut-free proof-search procedure for sequents containing formulas without tests (i.e.~programs of the form $\test{\varphi}$), and for which the relational atoms in the antecedents do not entail cyclic models.
We give a proof sketch that it is complete for this class of sequents.

Our approach relies on the following notion of normal form for sequents.
For a set of relational atoms and labelled formulas, we write $\starredlabs{\Gamma}$ for the set $\{ x \mid x : \nec{\iter{\alpha}}{\varphi} \in \Gamma \}$.
We call formulas of the form $\nec{a}{}\varphi$ \emph{basic}, those of the form $\nec{\iter{\alpha}}{\varphi}$ \emph{iterated}, and the remaining non-atomic formulas \emph{composite}.

\begin{definition}[Normal Sequents]
\label{def:NormalSequents}
  A sequent $\sequent{\Gamma}{\Delta}$ is called \emph{normal} when:
  \begin{enumerate*}[label={(\arabic*)}]
    \item
    $\Gamma \cap \Delta = \emptyset$;
    \item
    $\Delta$ contains only labelled atomic and iterated formulas; and
    \item
    $\Gamma$ contains only relational atoms, labelled atomic formulas, and labelled basic formulas $x : \nec{a}{\varphi}$ for which there is no $y$ such that also $x \relatedBy{a} y \in \Gamma$.
  \end{enumerate*}
\end{definition}

We say that $x$ \emph{reaches} $y$ (or $y$ is \emph{reachable} from $x$) in $\Gamma$ when there are labels $z_1, \ldots, z_n$ and atomic programs $a_1, \ldots, a_{n-1}$ such that $x = z_1$ and $y = z_n$ with $z_i \relatedBy{a_i} z_{n+1} \in \Gamma$ for each $i < n$.
We say that a sequent $\sequent{\Gamma}{\Delta}$ is \emph{cyclic} if there is some $x \in \labs{\Gamma}$ such that $x$ reaches itself in $\Gamma$; otherwise it is called \emph{acyclic}.

Crucially, the following forms of weakening are validity-preserving.

\begin{restatable}[Validity-preserving Weakenings]{lemma}{ValidWeakening}
\label{lem:ValidWeakening}
  The following hold.
  \begin{lemmaenum}[nosep,ref={{\thelemma}({\arabic*})}]
    \item
    \label{lem:ValidWeakening:RelationalAtoms:Right}
    If $\sequent{\Gamma}{\Delta, x \relatedBy{a} z}$ is valid and $x \relatedBy{a} z \not\in \Gamma$, then $\sequent{\Gamma}{\Delta}$ is valid.
    \item
    \label{lem:ValidWeakening:Formulas:Right}
    If normal $\sequent{\Gamma}{\Delta, x : p}$ is valid with $x \not\in \starredlabs{\Delta}$, then $\sequent{\Gamma}{\Delta}$ is valid.
    \item
    \label{lem:ValidWeakening:Formulas:Left}
    If normal $\sequent{\Gamma, x : \varphi}{\Delta}$ is valid with $x \not\in \labs{\Delta}$, then $\sequent{\Gamma}{\Delta}$ is valid.
    \item
    \label{lem:ValidWeakening:RelationalAtoms:Left:FreeUnreachable}
    If normal $\sequent{\Gamma, x \relatedBy{a} y}{\Delta}$ is valid, $z \in \labs{\Delta}$ for all $z : \varphi \in \Gamma$, $x \not\in \labs{\Delta}$ and $x$ not reachable in $\Gamma$ from any $z \in \labs{\Delta}$, then $\sequent{\Gamma}{\Delta}$ is valid.
  \end{lemmaenum}
\end{restatable}

An \emph{unwinding} of a sequent $\sequent{\Gamma}{\Delta}$ is a possibly open derivation of $\sequent{\Gamma}{\Delta}$ obtained by applying axioms and left and right logical rules as much as possible, and satisfying the properties that: no trace progresses more than once; and all rule instances consume the active labelled formula of their conclusion, but preserve in the premise any relational atoms.
A \emph{capped} unwinding is an unwinding for which:
\begin{enumerate*}[label={(\alph*)}]
  \item
  weakening rules and (Ax) and ($\falsum$) have been applied to all open leaves $\sequent{\Gamma}{\Delta}$ with $\falsum \in \Gamma$ or $\Gamma \cap \Delta \neq \emptyset$; and
  \item
  the sequence of weakenings in \cref{lem:ValidWeakening} have been exhaustively applied to all other open leaves.
\end{enumerate*}

\begin{restatable}{lemma}{UnwindingProperties}
  Let $\mathcal{D}$ be a capped unwinding for $\sequent{\Gamma}{\Delta}$ (denoted $S$) and $\sequent{\Gamma'}{\Delta'}$ an open leaf (denoted $S'$) of $\mathcal{D}$.
  The following hold:
  \begin{lemmaenum*}[nosep,label={({\arabic*})},ref={{\thelemma}(\arabic*)}]
    \item
    \label{lem:UnwindingProperties:NormalLeaves}
    $\sequent{\Gamma'}{\Delta'}$ is normal;
    \item
    \label{lem:UnwindingProperties:ValidLeaves}
    if $\sequent{\Gamma}{\Delta}$ is valid, then so are all the open leaves of $\mathcal{D}$; and
    \item
    \label{lem:UnwindingProperties:Traces}
    For every trace $\vec{\tau}_n$ covering the path from $S$ to $S'$, if $\tau_1 = (x, \emptysequence, \beta, \varphi)$ is a sub-formula of $\tau_n$, then the trace is progressing.
  \end{lemmaenum*}
\end{restatable}

We call a sequent \emph{test-free} if it does not contain any programs of the form $\test{\varphi}$.
A crucial property for termination of the proof-search is the following.

\begin{restatable}{lemma}{FiniteUnwindings}
\label{lem:FiniteUnwindings}
  Let $\mathcal{D}$ be a capped unwinding for a test-free, acyclic sequent; then $\mathcal{D}$ is finite, and $\labs{\Gamma'} \subseteq \labs{\Delta'} \subseteq \starredlabs{\Delta'}$ for all open leaves $\sequent{\Gamma'}{\Delta'}$ of $\mathcal{D}$.
\end{restatable}

Both cyclicity and the presence of tests can cause \cref{lem:FiniteUnwindings} to fail, since then it is possible for there to be a path of ancestry between two occurrences of an antecedent formula $x : \nec{\iter{\alpha}}{\varphi}$ that traverses an instance of the ($\ast$L) rule.
That is, antecedent formulas may be infinitely unfolded.
Moreover, in the presence of tests or cyclicity, the weakenings of \cref{lem:ValidWeakening:RelationalAtoms:Left:FreeUnreachable} do not result in $\labs{\Gamma'} \subseteq \labs{\Delta'}$ for open leaves $\sequent{\Gamma'}{\Delta'}$.

We define a function $\starMax$ on test-free sequents (details are given in the appendix), whose purpose is to provide a bound ensuring termination of proof-search.
Although, at time of submission, we do not have a fully detailed proof, we believe that it satisfies the following property for capped unwindings $\mathcal{D}$ of test-free, acyclic sequents $\sequent{\Gamma}{\Delta}$: $\left| \{ x : \varphi \in \Delta' \mid \text{$\varphi$ non-atomic} \} \right| \leq \starMax(\sequent{\Gamma}{\Delta})$ and $\starMax(\sequent{\Gamma'}{\Delta'}) \leq \starMax(\sequent{\Gamma}{\Delta})$ for all open leaves $\sequent{\Gamma'}{\Delta'}$ of $\mathcal{D}$.

Proof-search proceeds by iteratively building capped unwindings for open leaves.
All formulas encountered in the search are in the (finite) Fischer-Ladner closure of the initial sequent, and validity and acyclicity are preserved throughout the procedure.
\Cref{lem:FiniteUnwindings} and the above property will ensure that the number of distinct open leaves (modulo relabelling) encountered during proof-search is bounded, so we may apply substitutions to form back-links during proof-search.
\Cref{lem:UnwindingProperties:Traces} ensures that the resulting pre-proof satisfies the global trace condition.
For invalid sequents, proof-search produces atomic sequents that are not axioms. 
We thus conjecture cut-free regular completeness for test-free PDL.

\begin{restatable}{conjecture}{CutfreeRegularCompleteness}
  If test-free $\varphi$ is valid, $\sequent{}{x : \varphi}$ has a cut-free {\CyclicSystem} proof.
\end{restatable}

\section{Conclusion}

In this paper we have given two new non-wellfounded proof systems for PDL. 
{\InfinitarySystem} allows proof trees to be infinitely tall, and {\CyclicSystem} restricts to the proofs of {\InfinitarySystem} that are finitely representable as cyclic graphs satisfying a trace condition.
Soundness and completeness of both systems was shown, in particular, cut-free completeness
of {\InfinitarySystem} and a strategy for cut-free completeness of {\CyclicSystem} for test-free PDL.

There is much further work to be done. Of immediate interest is the verification of 
cut-free regular completeness for test-free PDL, and the extension of the argument to the full logic.
We would also like to consider \emph{additional} program constructs. Some, like converse, 
can already be treated through De Giacomo's \cite{DeGiacomo1996} efficient 
translation of Converse PDL into PDL.
It may be more desirable, however, to represent the program construct directly, 
to aid in the modular combination of different constructs. One construct that is 
particularly notorious is Intersection. Despite the modal definability of its dual, Choice, 
the intended interpretation of Intersection is \emph{not} modally definable, 
and the completeness (and existence) of an axiomatisation for it remained open until Balbiani and Vakarelov \cite{Balbiani2003}.
An earlier, and significantly simpler, solution to this problem was the augmentation of PDL with nominals, denoted Combinatory DL \cite{Passy1991}. 
We conjecture that the presence of labels in our system
enables us to perform a similar trick, without contaminating the syntax of the logic itself.
However we should note that a key prerequisite of our soundness proof, namely the finite model property of PDL, no longer applies to PDL with intersection.
We therefore have the non-trivial task ahead of us of weakening this assumption.

Our work should be seen as a part of a wider program of research to give a uniform and modular 
proof theory for a larger group of modal logics, including what we have denoted PDL-type logics. 
One source of modularity and uniformity
is the existing Negri labelled system our calculi extend. 
This allows us to freely add proof rules corresponding to 
first-order frame axioms defining Kripke models. 
A wider class of modal logics than those directly covered by 
Negri's framework are those with accessibility relations that 
are defined to be wellfounded or arise as transitive closures 
of other accessibility relations (we note Negri is able to treat the specific 
case of G\"{o}del-L\"{o}b logic due to its special interpretation of 
$\Box$, but not the general class we describe). We believe an appropriate framework to
uniformly capture \emph{these} logics as well is cyclic labelled deduction.
We are encouraged in this pursuit by recent work of Cohen and Rowe \cite{CohenRowe18} in
which first-order logic with a transitive closure operator is given a cyclic proof theory. 
We may think of labelled deduction as a way of giving a proof theoretic analysis 
of the first-order theory of Kripke models and their modal satisfaction relations. 
Labelled cyclic deduction, we conjecture, can be seen as the first-order-with-least-fixpoint 
theory of Kripke models and satisfaction relations.

Finally, and somewhat more speculatively, with the cyclic system in hand we intend to investigate the hitherto open problem of interpolation for PDL. 
This has seen no satisfactory resolution in the years since PDL was first formulated, with the only attempted proofs strongly disputed \cite{Kracht1999} or withdrawn \cite{Kowalski2004}. It would be interesting to see if the existence of a straightforward proof system for the logic opens up any new lines of attack on the problem. For example, Lyndon interpolation has been proved for G\"{o}del-L\"{o}b logic using a cyclic system \cite{Shamkanov2014}.

\clearpage
\bibliographystyle{plain}
\bibliography{cyclic_pdl}

\clearpage
\appendix

\section*{Appendix: Proofs}

\subsection*{Soundness}

Before proving \cref{lem:DescendingCountermodels}, we recall the following property of $<_{\text{DM}}$ (cf.~\cite[{\textsection}8(C)]{HuetOppen80}), where we write $M(x)$ for the multiplicity of $x$ in the multiset $M$.

\begin{property}
\label{prop:DMOrdering}
  For finite multisets $M$ and $N$ consisting of elements of a partial order $(\mathcal{S}, <_{\mathcal{S}})$, $M <_{\text{DM}} N$ if and only if:
  \begin{itemize}
    \item
    $M \neq N$; and
    \item
    for all $y \in \mathcal{S}$, if $N(y) < M(y)$ then there is $x \in \mathcal{S}$ such that $y <_{\mathcal{S}} x$ and $M(x) < N(x)$.
  \end{itemize}
\end{property}

We therefore have the following properties.
We call a function $f$ from a set $\Pi$ of paths in $\model$ for $\tau$ at $\valuation$ to another set $\Pi'$ of paths in $\model'$ for $\tau'$ at $\valuation'$ \emph{size-increasing} when it satisfies $\sizeOf[(\model, \valuation)]{\vec{s}}{\tau} \leq \sizeOf[(\model', \valuation')]{f(\vec{s})}{\tau'}$ for all $\vec{s} \in \Pi$, and \emph{strictly} size-increasing when $\sizeOf[(\model, \valuation)]{\vec{s}}{\tau} < \sizeOf[(\model', \valuation')]{f(\vec{s})}{\tau'}$ for at least one $\vec{s} \in \Pi$.

\begin{proposition}
  Let $\Pi$ and $\Pi'$ be non-empty finite sets of paths in $\model$ for $\tau$ and $\tau'$ at $\valuation$ and $\valuation'$, respectively, and let $f : \Pi \rightarrow \Pi'$ be injective.
  Then the following hold.
  \begin{enumerate}[ref={{\theproposition}(\theenumi)}]
    \item
    \label[proposition]{lem:DMOrdering:SizeIncreasing:Strict}
    If $f$ is strictly size-increasing then $\pathsetMeasure{\Pi} <_{\text{DM}} \pathsetMeasure{\Pi'}$.
    \item
    \label[proposition]{lem:DMOrdering:SizeIncreasing:Nonstrict}
    If $f$ is size-increasing then $\pathsetMeasure{\Pi} \leq_{\text{DM}} \pathsetMeasure{\Pi'}$.
  \end{enumerate} 
\end{proposition}
\begin{proof}
  \begin{enumerate}[label={(\arabic*)}]
    \item
    Notice that since $\Pi$ is finite there is some maximal $i$ such that there is some path $\vec{s} \in \Pi$ with $i = \sizeOf[(\model, \valuation)]{\vec{s}}{\tau} < \sizeOf[(\model, \valuation')]{f(\vec{s})}{\tau'}$ (i.e.~a greatest path weight in $\Pi$ that is strictly increased by $f$).
    Therefore there is some $j > i$ such that $\pathsetMeasure{\Pi'}(j) > \pathsetMeasure{\Pi}(i)$---simply take $j = \sizeOf[(\model', \valuation')]{f(\vec{s})}{\tau'}$ for some $\vec{s} \in \Pi$ such that $i = \sizeOf[(\model, \valuation)]{\vec{s}}{\tau}$---and so $\pathsetMeasure{\Pi} \neq \pathsetMeasure{\Pi'}$.
    Moreover, notice that $\pathsetMeasure{\Pi}(k) \leq \pathsetMeasure{\Pi'}(k)$ for all $k > i$.
    Thus, if we have $k$ such that $\pathsetMeasure{\Pi'}(k) < \pathsetMeasure{\Pi}(k)$ it must be that $k \leq i$.
    But then the second condition of \cref{prop:DMOrdering} is satisfied since we have $\pathsetMeasure{\Pi}(j) < \pathsetMeasure{\Pi'}(j)$ and $k \leq i < j$.
    Thus $\pathsetMeasure{\Pi} <_{\text{DM}} \pathsetMeasure{\Pi'}$.
    \item
    We distinguish two cases.
    \begin{itemize}
      \item
      If, in fact, $\sizeOf[(\model, \valuation)]{\vec{s}}{\tau} = \sizeOf[(\model', \valuation')]{f(\vec{s})}{\tau'}$ for all $\vec{s} \in \Pi$ then we distinguish two further sub-cases.
      On the one hand, if $f$ is surjective then we have that $\pathsetMeasure{\Pi} = \pathsetMeasure{\Pi'}$, and therefore immediately $\pathsetMeasure{\Pi} \leq_{\text{DM}} \pathsetMeasure{\Pi'}$.
      On the other hand, when $f$ is not surjective then we have $\pathsetMeasure{\Pi'} \neq \pathsetMeasure{\Pi}$ (since $\Pi'$ contains more paths) but $\pathsetMeasure{\Pi}(i) \leq \pathsetMeasure{\Pi'}(i)$, for all $i \geq 0$; thus the second condition of \cref{prop:DMOrdering} is trivially satisfied and so in fact $\pathsetMeasure{\Pi} <_{\text{DM}} \pathsetMeasure{\Pi'}$.
      Therefore also $\pathsetMeasure{\Pi} \leq_{\text{DM}} \pathsetMeasure{\Pi'}$.
      \item
      If there is some $\vec{s} \in \Pi$ such that $\sizeOf[(\model, \valuation)]{\vec{s}}{\tau} < \sizeOf[(\model', \valuation')]{f(\vec{s})}{\tau'}$, then $f$ is in fact strictly size-increasing.
      Thus we have from the previous result that $\pathsetMeasure{\Pi} <_{\text{DM}} \pathsetMeasure{\Pi'}$, and so also $\pathsetMeasure{\Pi} \leq_{\text{DM}} \pathsetMeasure{\Pi'}$.
      \qed
    \end{itemize}
  \end{enumerate}
\end{proof}

We now prove the descending counter-models lemma.

\DescendingCountermodels*
\begin{proof}
  Let $S$ be the conclusion of an instance of an inference rule; suppose $\model$ is a finitely branching model and $\valuation$ such that $\model, \valuation \not\models S$.
  We do a case analysis on the inference rule.
  Recall that: trace value measures are always finite for finitely branching models; and the set of counter-examples for any trace value in an invalid sequent is necessarily non-empty.
  \begin{description}[font={\normalfont},itemsep={\smallskipamount},listparindent={\parindent}]
    \item[(Ax), ($\falsum$):]
    The result holds vacuously since it is clear that these rules derive valid sequents.

    \item[(Subst):]
    Straightforward.
    We take $\valuation' = \valuation [y := \valuation(x)]$.
    For any trace pair $(\tau, \tau')$, the sets of counter-examples in $\model$ for $\tau$ and $\tau'$ at $\valuation$ and $\valuation'$, respectively, are equal.
    Thus we have $\measureOf{\model}{\valuation'}{\tau'} = \measureOf{\model}{\valuation}{\tau}$ and so trivially $\measureOf{\model}{\valuation'}{\tau'} \leq_{\text{DM}} \measureOf{\model}{\valuation}{\tau}$.

    \item[(Classical connectives, left modal rules, Cut):]
    These cases are straightforward: taking $\valuation' = \valuation$ suffices.
    Moreover, in these cases each trace pair $(\tau, \tau')$ is such that $\tau = \tau'$ (and therefore non-progressing); thus the trace condition holds straightforwardly because then the sets of counter-examples in $\model$ for $\tau$ and $\tau'$ at $\valuation$ and $\valuation'$, respectively, are equal.
    Thus we have $\measureOf{\model}{\valuation'}{\tau'} = \measureOf{\model}{\valuation}{\tau}$ and so trivially $\measureOf{\model}{\valuation'}{\tau'} \leq_{\text{DM}} \measureOf{\model}{\valuation}{\tau}$.

    \item[($\square$R):]
    Then $S$ is $\sequent{\Gamma}{\Delta, x : \nec{a}{\varphi}}$ and $S'$ is of the form $\sequent{x \relatedBy{a} y, \Gamma}{\Delta, y : \varphi}$ for some fresh $y$.
    Since $\model, \valuation \not\models x : \nec{a}{\varphi}$ we have that there is some state $s$ in $\model$ such that $(\valuation(x), s) \in \interp[\model]{a}$ and $\model, s \not\models \varphi$.
    So we take $\valuation' = \valuation [y := s]$, for which we have $\model, \valuation' \not\models S'$.
    For all (non-progressing) trace pairs $(\tau, \tau')$ in the context $\Delta$, we have $\tau = \tau'$ with label $z \neq y$, and so $\valuation'(z) = \valuation(z)$.
    Thus $\measureOf{\model}{\valuation'}{\tau'} = \measureOf{\model}{\valuation}{\tau}$ and therefore trivially $\measureOf{\model}{\valuation'}{\tau'} \leq_{\text{DM}} \measureOf{\model}{\valuation}{\tau}$.
    For trace pairs $(\tau, \tau')$ where $\tau' = (y, \vec{\alpha}, \beta, \psi)$ and $\tau = y : \nec{a}{\tau}$, we reason as follows.
    Notice that we have $\vec{s} \in \counterEgs{\model}{\valuation'}{\tau'}$ only if $s \cdot \vec{s} \in \counterEgs{\model}{\valuation}{\tau}$.
    That is, $f$ defined by $f(\vec{s}) = s \cdot \vec{s}$ is an injection from $\counterEgs{\model}{\valuation'}{\tau'}$ to $\counterEgs{\model}{\valuation}{\tau}$.
    Notice also that if $\vec{k}_n$ is a partition for $\vec{s} \in \counterEgs{\model}{\valuation'}{\tau'}$ then $2, k_1 + 1, \ldots, k_n + 1$ is a partition for $s \cdot \vec{s} \in \counterEgs{\model}{\valuation}{\tau}$.
    Therefore, $\sizeOf[(\model, \valuation')]{\vec{s}}{\tau'} \leq \sizeOf[(\model, \valuation)]{s \cdot \vec{s}}{\tau}$.
    To see this, take some maximal weight partition $\vec{k}$ of $\vec{s} \in \counterEgs{\model}{\valuation'}{\tau'}$, with weight $\ell$; then we know that $2, k_1 + 1, \ldots, k_n + 1$ is a partition for $s \cdot \vec{s} \in \counterEgs{\model}{\valuation}{\tau}$.
    The weight of $2, k_1 + 1, \ldots, k_n + 1$ is clearly also $\ell$, and so the weight of the maximal partition of $\vec{s}$ for $\tau$ at $\valuation$ must be at least $\ell$. 
    In summary, $f$ is size-increasing.
    So, by \cref{lem:DMOrdering:SizeIncreasing:Nonstrict}, $\measureOf{\model}{\valuation'}{\tau'} \leq_{\text{DM}} \measureOf{\model}{\valuation}{\tau}$.

    \item[($\compSymb$R):]
    Then $S$ is $\sequent{\Gamma}{\Delta, x : \nec{\comp{\alpha}{\beta}}{\varphi}}$ and $S'$ is $\sequent{\Gamma}{\Delta, x : \nec{\alpha}{\nec{\beta}{\varphi}}}$.
    We take $\valuation' = \valuation$ since, if $s$ is a state in $\model$ such that $(\valuation(x), s) \in \interp[\model]{\comp{\alpha}{\beta}}$ and $\model, s \not\models \varphi$, then by \cref{def:Semantics} there is some state $s'$ in $\model$ with $(\valuation(x), s') \in \interp[\model]{\alpha}$ and $(s', s) \in \interp[\model]{\beta}$.
    For all (non-progressing) trace pairs $(\tau, \tau')$ in the context $\Delta$, we have $\tau = \tau'$; so $\measureOf{\model}{\valuation'}{\tau'} = \measureOf{\model}{\valuation}{\tau}$ and therefore $\measureOf{\model}{\valuation'}{\tau'} \leq_{\text{DM}} \measureOf{\model}{\valuation}{\tau}$ holds trivially.
    For (non-progressing) trace pairs $(\tau, \tau')$ where $\tau = \nec{\comp{\alpha}{\beta}}{\tau''}$ and $\tau' = \nec{\alpha}{\nec{\beta}{\tau''}}$ for some $\tau''$ we reason as follows.
    Notice that we have $\vec{s} \in \counterEgs{\model}{\valuation'}{\tau'}$ (if and) only if $\vec{s} \in \counterEgs{\model}{\valuation}{\tau}$.
    Moreover, if $j \cdot \vec{k}_n$ is a partition for $\vec{s} \in \counterEgs{\model}{\valuation'}{\tau'}$ then $\vec{k}_n$ is a partition for $\vec{s} \in \counterEgs{\model}{\valuation}{\tau}$.
    Therefore, $\sizeOf[(\model, \valuation')]{\vec{s}}{\tau'} \leq \sizeOf[(\model, \valuation)]{\vec{s}}{\tau}$.
    To see this, take some maximal weight partition $j \cdot \vec{k}$ of $\vec{s} \in \counterEgs{\model}{\valuation'}{\tau'}$, with weight $\ell$; then we know that $\vec{k}$ is a partition for $\vec{s} \in \counterEgs{\model}{\valuation}{\tau}$.
    The weight of $\vec{k}$ is clearly also $\ell$, and so the weight of the maximal partition of $\vec{s}$ for $\tau$ at $\valuation$ must be at least $\ell$. 
    In summary, the identity function is a size-increasing injection from $\counterEgs{\model}{\valuation'}{\tau'}$ to $\counterEgs{\model}{\valuation}{\tau}$.
    So, by \cref{lem:DMOrdering:SizeIncreasing:Nonstrict}, $\measureOf{\model}{\valuation'}{\tau'} \leq_{\text{DM}} \measureOf{\model}{\valuation}{\tau}$.

    \item[($\choiceSymb$R):]
    Then $S$ is $\sequent{\Gamma}{\Delta, x : \nec{\choice{\alpha}{\beta}}{\varphi}}$.
    We take $\valuation' = \valuation$.
    If $s$ is a state in $\model$ such that $(\valuation(x), s) \in \interp[\model]{\choice{\alpha}{\beta}}$ and $\model, s \not\models \varphi$ then, by \cref{def:Semantics}, either $s \in \interp[\model]{\alpha}$ or $s \in \interp[\model]{\beta}$.
    If the former we take the left-hand premise for $S'$, and if the latter then we take the right-hand premise.
    In both cases, notice that we can reason as follows.
    Firstly, since $\valuation' = \valuation$, for all (non-progressing) trace pairs $(\tau, \tau')$ in the context $\Delta$, we have $\tau = \tau'$; so $\measureOf{\model}{\valuation'}{\tau'} = \measureOf{\model}{\valuation}{\tau}$ and therefore trivially $\measureOf{\model}{\valuation'}{\tau'} \leq_{\text{DM}} \measureOf{\model}{\valuation}{\tau}$.
    Secondly, for (non-progressing) trace pairs $(\tau, \tau')$ where $\tau = \nec{\choice{\alpha}{\beta}}{\tau''}$, and $\tau' = \nec{\gamma}{\tau''}$ (with $\gamma = \alpha$ or $\gamma = \beta$) for some $\tau''$, we reason as follows.
    Notice that we have $\vec{s} \in \counterEgs{\model}{\valuation'}{\tau'}$ only if $\vec{s} \in \counterEgs{\model}{\valuation}{\tau}$.
    Moreover, if $\vec{k}_n$ is a partition for $\vec{s} \in \counterEgs{\model}{\valuation'}{\tau'}$ then $\vec{k}_n$ is also a partition for $\vec{s} \in \counterEgs{\model}{\valuation}{\tau}$.
    Therefore, $\sizeOf[(\model, \valuation')]{\vec{s}}{\tau'} \leq \sizeOf[(\model, \valuation)]{\vec{s}}{\tau}$.
    To see this, take some maximal weight partition $\vec{k}$ of $\vec{s} \in \counterEgs{\model}{\valuation'}{\tau'}$, with weight $\ell$; then we know that $\vec{k}$ is also a partition for $\vec{s} \in \counterEgs{\model}{\valuation}{\tau}$.
    Thus the weight of the maximal partition of $\vec{s}$ for $\tau$ at $\valuation$ must be at least $\ell$. 
    In summary, the identity function is a size-increasing injection from $\counterEgs{\model}{\valuation'}{\tau'}$ to $\counterEgs{\model}{\valuation}{\tau}$.
    So, by \cref{lem:DMOrdering:SizeIncreasing:Nonstrict}, $\measureOf{\model}{\valuation'}{\tau'} \leq_{\text{DM}} \measureOf{\model}{\valuation}{\tau}$.

    \item[($\testSymb$R):]
    Then $S$ is $\sequent{\Gamma}{\Delta, x : \nec{\test{\varphi}}{\psi}}$ and $S'$ is $\sequent{x : \varphi, \Gamma}{\Delta, x : \psi}$.
    We take $\valuation' = \valuation$ since, if $s$ is a state in $\model$ with $(\valuation(x), s) \in \interp[\model]{\test{\varphi}}$ such that $\model, s \not\models \psi$ then, by \cref{def:Semantics}, $s = \valuation(x)$ and $\model, s \models \varphi$.
    Thus, for all (non-progressing) trace pairs $(\tau, \tau')$ in the context $\Delta$, we have $\tau = \tau'$; so $\measureOf{\model}{\valuation'}{\tau'} = \measureOf{\model}{\valuation}{\tau}$ and therefore trivially $\measureOf{\model}{\valuation'}{\tau'} \leq_{\text{DM}} \measureOf{\model}{\valuation}{\tau}$.
    For (non-progressing) trace pairs $(\tau, \tau')$, we have $\tau = \nec{\test{\varphi}}{\tau'}$, and we can reason as follows.
    Notice that we have $\vec{s} \in \counterEgs{\model}{\valuation'}{\tau'}$ (if and) only if $\vec{s} \in \counterEgs{\model}{\valuation}{\tau}$.
    Moreover, if $\vec{k}_n$ is a partition for $\vec{s} \in \counterEgs{\model}{\valuation'}{\tau'}$ then $1 \cdot \vec{k}_n$ is a partition for $\vec{s} \in \counterEgs{\model}{\valuation}{\tau}$.
    Therefore, $\sizeOf[(\model, \valuation')]{\vec{s}}{\tau'} \leq \sizeOf[(\model, \valuation)]{\vec{s}}{\tau}$.
    To see this, take some maximal weight partition $\vec{k}$ of $\vec{s} \in \counterEgs{\model}{\valuation'}{\tau'}$, with weight $\ell$; then we know that $1 \cdot \vec{k}$ is also a partition for $\vec{s} \in \counterEgs{\model}{\valuation}{\tau}$.
    The weight of $1 \cdot \vec{k}$ is clearly also $\ell$, and so the weight of the maximal partition of $\vec{s}$ for $\tau$ at $\valuation$ must be at least $\ell$. 
    In summary, the identity function is a size-increasing injection from $\counterEgs{\model}{\valuation'}{\tau'}$ to $\counterEgs{\model}{\valuation}{\tau}$.
    So, by \cref{lem:DMOrdering:SizeIncreasing:Nonstrict}, $\measureOf{\model}{\valuation'}{\tau'} \leq_{\text{DM}} \measureOf{\model}{\valuation}{\tau}$.

    \item[($\ast$R):]
    Then $S$ is $\sequent{\Gamma}{\Delta, x : \nec{\iter{\alpha}}{\varphi}}$.
    We take $\valuation' = \valuation$.
    Since $\model, \valuation \not\models S$, there is a state $s$ in $\model$ with $(\valuation(x), s) \in \interp[\model]{\iter{\alpha}} = \bigcup_{k \geq 0} {\interp[\model]{\alpha}}^k$ such that $\model, s \not\models \varphi$.
    Then, there are two possibilities.
    If $\model, \valuation(x) \not\models \varphi$ then we take the left-hand premise and so $S'$ is $\sequent{\Gamma}{\Delta, x : \varphi}$.
    On the other hand, if in fact $\model, \valuation(x) \models \varphi$ then $s \neq \valuation(x)$ and we have that $(\valuation(x), s) \in {\interp[\model]{\alpha}}^k$ for some $k > 0$.
    Therefore there is a state $s' \neq \valuation(x)$ in $\model$ with $(\valuation(x), s') \in \interp[\model]{\alpha}$ and $(s', s) \in \interp[\model]{\iter{\alpha}}$.
    In this case, the right-hand premise is invalidated and so we take $S'$ is $\sequent{\Gamma}{\Delta, x : \nec{\alpha}{\nec{\iter{\alpha}}{\varphi}}}$.
    In both cases, for the (non-progressing) pairs $(\tau, \tau')$ in the context $\Delta$ we have $\tau = \tau'$; so $\measureOf{\model}{\valuation'}{\tau'} = \measureOf{\model}{\valuation}{\tau}$ and so trivially $\measureOf{\model}{\valuation'}{\tau'} \leq_{\text{DM}} \measureOf{\model}{\valuation}{\tau}$.
    For the remaining trace pairs, there are two cases.
    \begin{itemize}[topsep={\smallskipamount},itemsep={\smallskipamount}]
      \item
      If $S'$ is the left-hand premise, then we reason as follows.
      In this case $\tau = \nec{\iter{\alpha}}{\tau'}$. 
      Notice that $\vec{s} \in \counterEgs{\model}{\valuation'}{\tau'}$ only if $\vec{s} \in \counterEgs{\model}{\valuation}{\tau}$.
      Moreover, if $\vec{k}_n$ is a partition for $\vec{s} \in \counterEgs{\model}{\valuation'}{\tau'}$ then $1 \cdot \vec{k}_n$ is a partition for $\vec{s} \in \counterEgs{\model}{\valuation}{\tau}$.
      Therefore, $\sizeOf[(\model, \valuation')]{\vec{s}}{\tau'} \leq \sizeOf[(\model, \valuation)]{\vec{s}}{\tau}$.
      To see this, take some maximal weight partition $\vec{k}$ of $\vec{s} \in \counterEgs{\model}{\valuation'}{\tau'}$, with weight $\ell$; then we know that $1 \cdot \vec{k}$ is also a partition for $\vec{s} \in \counterEgs{\model}{\valuation}{\tau}$.
      The weight of $1 \cdot \vec{k}$ is clearly also $\ell$, and so the weight of the maximal partition of $\vec{s}$ for $\tau$ at $\valuation$ must be at least $\ell$. 
      In summary, the identity function is a size-increasing injection from $\counterEgs{\model}{\valuation'}{\tau'}$ to $\counterEgs{\model}{\valuation}{\tau}$.
      So, by \cref{lem:DMOrdering:SizeIncreasing:Nonstrict}, $\measureOf{\model}{\valuation'}{\tau'} \leq_{\text{DM}} \measureOf{\model}{\valuation}{\tau}$.
      \item
      If $S'$ is the right-hand premise, then recall we have $\model, \valuation(x) \models \varphi$ and so $s \neq s'$ for all $(s, s') \in \interp[\model]{\alpha}$.
      Here we have $\tau' = \nec{alpha}{\tau}$, with $(\tau, \tau')$ progressing if and only if the spine of $\tau$ is empty.
      Notice that, for this case, we have $\vec{s} \in \counterEgs{\model}{\valuation'}{\tau'}$ (if and) only if $\vec{s} \in \counterEgs{\model}{\valuation}{\tau}$.
      Moreover, if $j \cdot \vec{k}_n$ is a partition for $\vec{s} \in \counterEgs{\model}{\valuation'}{\tau'}$ then $\vec{k}_n$ is a partition for $\vec{s} \in \counterEgs{\model}{\valuation}{\tau}$.
      Therefore, $\sizeOf[(\model, \valuation')]{\vec{s}}{\tau'} \leq \sizeOf[(\model, \valuation)]{\vec{s}}{\tau}$.
      To see this, take some maximal weight partition $j \cdot \vec{k}$ of $\vec{s} \in \counterEgs{\model}{\valuation'}{\tau'}$, with weight $\ell$; then we know that $\vec{k}$ is also a partition for $\vec{s} \in \counterEgs{\model}{\valuation}{\tau}$.
      The weight of $\vec{k}$ is clearly also $\ell$, and so the weight of the maximal partition of $\vec{s}$ for $\tau$ at $\valuation$ must be at least $\ell$. 
      In summary, the identity function is a size-increasing injection from $\counterEgs{\model}{\valuation'}{\tau'}$ to $\counterEgs{\model}{\valuation}{\tau}$.
      So, by \cref{lem:DMOrdering:SizeIncreasing:Nonstrict}, $\measureOf{\model}{\valuation'}{\tau'} \leq_{\text{DM}} \measureOf{\model}{\valuation}{\tau}$.
      Now, as noted above $s \neq s'$ for all $(s, s') \in \interp[\model]{\alpha}$, and so for any partition $j \cdot \vec{k}$ of $\vec{s}_n \in \counterEgs{\model}{\valuation'}{\tau'}$ we have that $j > 1$ and so $\sizeOf[(\model, \valuation')]{\vec{s}}{\tau'} < n$.
      Thus, in the case that the spine of $\tau$ is empty, we have $\sizeOf[(\model, \valuation')]{\vec{s}}{\tau'} < n = \sizeOf[(\model, \valuation)]{\vec{s}}{\tau}$.
      That is, for the progressing trace pair, the identity function is a strictly size-increasing injection from $\counterEgs{\model}{\valuation'}{\tau'}$ to $\counterEgs{\model}{\valuation}{\tau}$ and so, by \cref{lem:DMOrdering:SizeIncreasing:Strict}, $\measureOf{\model}{\valuation'}{\tau'} <_{\text{DM}} \measureOf{\model}{\valuation}{\tau}$.
      \qed
    \end{itemize}
  \end{description}
\end{proof}

\Soundness*
\begin{proof}
  Take a {\InfinitarySystem} proof $\mathcal{P}$ of a sequent $\sequent{\Gamma}{\Delta}$ and assume for contradiction that it is not valid.
  Since PDL has the finite model property, there is a finitely branching model $\model$ and $\valuation$ that falsifies it.
  Then, by \cref{lem:DescendingCountermodels}, $\mathcal{P}$ must contain an infinite path $\sequent{\Gamma_1}{\Delta_1}, \sequent{\Gamma_2}{\Delta_2}, \ldots$ for which there exists a sequence of valuations $\valuation_1, \valuation_2, \ldots$ that invalidate each sequent in the path.
  Moreover, these invalidating valuations give rise to measures for the trace values that ensure the measure of any trace pair along the path is decreasing, and strictly so for progressing trace pairs.
  However, since $\mathcal{P}$ is a proof, it satisfies the global trace condition.
  Thus there is an infinitely progressing trace along the path $\sequent{\Gamma_1}{\Delta_1}, \sequent{\Gamma_2}{\Delta_2}, \ldots$ and therefore we can construct an infinitely descending chain of (finite) trace value measures.
  But this contradicts the fact that the set of finite trace value measures is well-founded.
  So conclude that $\sequent{\Gamma}{\Delta}$ must in fact be valid.
  \qed
\end{proof}

\subsection*{Completeness}

\CanonicalModels*
\begin{proof}
  By induction on the structure of formulas.
  \begin{description}[font={\normalfont},listparindent={\parindent}]
    \item[($\falsum$):]
    The first conjunct follows trivially since, by construction of the search tree, $x : \falsum$ cannot be in $\Gamma$.
    If $x : \falsum \in \Delta$ then the result follows immediately since $\interp[\model_{P}]{\falsum} = \emptyset$ by \cref{def:Semantics}.
    \item[(relational atom, or atomic formula $p$):]
    Immediate by \cref{def:CanonicalModel}, since $\Gamma$ and $\Delta$ are disjoint.
    \item[($\wedge$, $\vee$, $\rightarrow$):]
    Straightforward induction.
    \item[(\nec{\alpha}{}):]
    The first conjunct is entailed by the following property, which we show for all programs $\beta$, labels $x$ and $y$, and formulas $\varphi$.
    \begin{gather*}
      \tag{A}
      \label[property]{prop:Template:LeftClosure}
      \begin{split}
      \text{If $\beta$ is a subprogram of $\alpha$, $x : \nec{\beta}{\varphi} \in \Gamma$ and $(x, y) \in \interp[\model_{P}]{\beta}$,}
        \hspace{2em} \\
      \text{then $y : \varphi \in \Gamma$}
      \end{split}
    \end{gather*}
    To see this, assume that $x : \nec{\alpha}{\varphi} \in \Gamma$.
    We need to show that $\model_{P}, \idValuation \models y : \varphi$ for all $y$ such that $(x, y) \in \interp[\model_{P}]{\alpha}$.
    But this follows immediately by induction, since we have from \cref{prop:Template:LeftClosure} that $y : \varphi \in \Gamma$ (note that the subprogram relation is reflexive).
    We prove \cref{prop:Template:LeftClosure} by an inner induction on the structure of programs.
    \begin{itemize}[wide,topsep={\medskipamount}]
      \item
      For $\beta$ an atomic program $a$, assume $x : \nec{a}{\varphi} \in \Gamma$ and $(x, y) \in \interp[\model_{P}]{a}$.
      From the latter, by \cref{def:CanonicalModel}, it must be that $x \relatedBy{a} y \in \Gamma$.
      Thus, the result follows from the construction of the search tree (cf.~\cref{def:SearchTree}).
      \item
      Composition and choice follow by a straightforward induction.
      \item
      For $\beta$ a test $\test{\varphi}$, notice that since $\beta$ is a subprogram of $\alpha$ we have that the outer induction applies to $\varphi$.
      Assume that $(x, y) \in \interp[\model_{P}]{\test{\varphi}}$.
      Thus by \cref{def:Semantics}, we have $x = y$ and $x \in \interp[\model_{P}]{\varphi}$.
      Assume also that $x : \nec{\test{\varphi}}{\psi} \in \Gamma$.
      By construction of the search tree, either $x : \psi \in \Gamma$ or $x : \varphi \in \Delta$.
      The result obtains since the latter cannot hold; else by the outer induction we have that $\model_{P}, \idValuation \not\models x : \varphi$, i.e.~$x \not\in \interp[\model_{P}]{\varphi}$.
      \item
      For iteration, assume $(x, y) \in \interp[\model_{P}]{\iter{\beta}}$.
      So, there are labels $z_0, \ldots, z_n$ with $x = z_0$ and $z_n = y$ such that $(z_i, z_{i+1}) \in \interp[\model_{P}]{\beta}$ for each $0 \leq i < n$.
      Now assume that $x : \nec{\iter{\beta}}{\varphi} \in \Gamma$.
      It suffices to show that if there are labels $z_0, \ldots, z_n$ with $(z_i, z_{i+1}) \in \interp[\model_{P}]{\beta}$ for each $0 \leq i < n$ and $z_0 : \nec{\iter{\beta}}{\varphi} \in \Gamma$, then both $z_n : \varphi \in \Gamma$ and $z_n : \nec{\beta}{\nec{\iter{\beta}}{\varphi}} \in \Gamma$.
      This follows by a further induction on $n$.
      For $n = 0$, the result follows immediately by construction of the search tree.
      For $n = k + 1$, $n_k : \nec{\beta}{\nec{\iter{\beta}}{\varphi}} \in \Gamma$ by the induction on $n$.
      Then by the induction on \cref{prop:Template:LeftClosure} we have that $z_n : \nec{\iter{\beta}}{\varphi} \in \Gamma$, whence the result follows again by construction of the search tree.
    \end{itemize}

    The second conjunct is entailed by the following property, which we again show for all programs $\beta$, labels $x$, and formulas $\varphi$.
    Let $\sequent{\Gamma_0}{\Delta_0}, \sequent{\Gamma_1}{\Delta_1}, \ldots$ be the untraceable branch from which the template $(\Gamma, \Delta)$ is derived.
    \begin{gather*}
      \tag{B}
      \label[property]{prop:Template:RightClosure}
      \begin{aligned}
        &
        \text{If $\beta$ is a subprogram of $\alpha$ and $x : \nec{\beta}{\varphi} \in \Delta_i$, then}
          \\[-0.25em] &\hspace{1em}
        \text{there exist $y$ and $j > i$ such that $y : \varphi \in \Delta_j$, $(x, y) \in \interp[\model_{P}]{\beta}$, and}
          \\[-0.25em] &\hspace{3em} 
        \text{for all trace values $\tau = y : \varphi$ there is a trace $x : \nec{\beta}{\tau}, \ldots, \tau$}
          \\[-0.25em] &\hspace{6em}
        \text{covering the path $\sequent{\Gamma_i}{\Delta_i}, \ldots, \sequent{\Gamma_j}{\Gamma_j}$}
      \end{aligned}
    \end{gather*}
    To see this, assume $x : \nec{\alpha}{\varphi} \in \Delta$; thus $x : \nec{\alpha}{\varphi} \in \Delta_i$ for some $i$.
    We need to show that $\model_{P}, \idValuation \not\models y : \varphi$ for some $y$ such that $(x, y) \in \interp[\model_{P}]{\alpha}$.
    But this follows immediately by induction, since we have from \cref{prop:Template:RightClosure} that $y : \varphi \in \Delta_j$ for some $y$ and $j$, and thus $y : \varphi \in \Delta$ by definition, with $(x, y) \in \interp[\model_{P}]{\alpha}$.
    Again, we prove \cref{prop:Template:RightClosure} by an inner induction on the structure of programs.
    The clause relating to traces is needed for iteration.
    \begin{itemize}[wide,topsep={\medskipamount}]
      \item
      For $\beta$ an atomic program $a$, assume $x : \nec{a}{\varphi} \in \Delta$.
      By construction of the search tree, $y : \varphi \in \Delta_j$ and $x \relatedBy{a} y \in \Gamma_j$ for some $j > k$, and there is a trace $x : \nec{a}{\tau}, \ldots, \tau$ covering the path $\sequent{\Gamma_i}{\Delta_i}, \ldots, \sequent{\Gamma_j}{\Delta_j}$ for every trace value $\tau = y : \varphi$.
      Thus, the result follows since then $(x, y) \in \interp[\model_{P}]{a}$ by \cref{def:CanonicalModel}.
      \item
      Composition and choice follow by a straightforward induction.
      \item
      For $\beta$ a test $\test{\varphi}$, notice that since $\beta$ is a subprogram of $\alpha$ we have that the outer induction applies to $\varphi$.
      Assume that $x : \nec{\test{\varphi}}{\psi} \in \Delta_i$.
      By construction of the search tree, we have that both $x : \varphi \in \Gamma_j$ and $x : \psi \in \Delta_j$ for some $j > k$, and there is a trace $\nec{\test{\varphi}}{\tau}, \ldots, \tau$ covering the path $\sequent{\Gamma_i}{\Delta_i}, \ldots, \sequent{\Gamma_j}{\Delta_j}$ for every trace value $\tau = x : \psi$.
      Thus, by the outer induction $\model_{P}, \idValuation \models x : \varphi$.
      So $(x, x) \in \interp[\model_{P}]{\test{\varphi}}$ by \cref{def:Semantics}, whence the result follows taking $y = x$.
      \item
      For iteration, assume $x : \nec{\iter{\beta}}{\varphi} \in \Delta_i$.
      By the construction of the search tree and the inner induction, we have for some $j > k$:
      \begin{enumerate*}[label={(\roman*)}]
        \item
        either $x : \varphi \in \Delta_j$ and there is a trace $\tau, \ldots, \tau$ covering the path $\sequent{\Gamma_i}{\Delta_i}, \ldots, \sequent{\Gamma_j}{\Delta_j}$ for every trace value $\tau = x : \varphi$; or
        \item
        $y : \nec{\iter{\beta}}{\varphi} \in \Delta_j$ for some $y$ with $(x, y) \in \interp[\model_{P}]{\beta}$ and there is a trace $\tau, \ldots, y : \tau$ covering the path $\sequent{\Gamma_i}{\Delta_i}, \ldots, \sequent{\Gamma_j}{\Delta_j}$ for every trace value $\tau = x : \nec{\iter{\beta}}{\varphi}$, which is progressing for $\tau = (x, \emptysequence, \beta, \varphi)$, where $\emptysequence$ is the empty sequence.
      \end{enumerate*}
      Therefore, there must be some $z_0, \ldots, z_n$ and $k > i$ with $z_0 = x$ such that
        $z_n : \varphi \in \Delta_k$,
        $(z_i, z_{i+1}) \in \interp[\model_{P}]{\beta}$ for each $i < n$, and for every trace value $\tau = z_n : \varphi$
        there is a trace $z_0 : \nec{\iter{\beta}}{\tau}, \ldots, \tau$ covering the path $\sequent{\Gamma_i}{\Delta_i}, \ldots, \sequent{\Gamma_k}{\Gamma_k}$.
      If not there would be an infinitely progressing trace covering the untraceable branch, which is impossible by definition.
      Thence the result follows, since $(x, z_n) \in \interp[\model_{P}]{\iter{\beta}}$ by \cref{def:Semantics}.
      \qed
    \end{itemize}
  \end{description}
\end{proof}

\begin{remark}
  It is only the presence of tests in programs that requires \cref{prop:Template:LeftClosure,prop:Template:RightClosure} above to be proved via a nested induction.
  Without tests, these properties may be proved independently.
\end{remark}

\CutfreeCompleteness*
\begin{proof}
  Suppose $\sequent{\Gamma}{\Delta}$ is valid.
  Let $\mathcal{D}$ be a search tree for $\sequent{\Gamma}{\Delta}$.
  It must be that $\mathcal{D}$ is a (cut-free) proof.
  Supposing otherwise, it would induce some template $P = (\Gamma', \Delta')$.
  Since, by construction, $\Gamma \subseteq \Gamma'$ and $\Delta \subseteq \Delta'$, it then follows from \cref{lem:CanonicalModel} that the model determined by $P$ satisfies all $A \in \Gamma$ and falsifies all $B \in \Delta$ (i.e.~$\model_{P}$ is a countermodel); thus $\sequent{\Gamma}{\Delta}$ is not valid.
  However, this contradicts our initial supposition that $\sequent{\Gamma}{\Delta}$ is valid and so $\mathcal{D}$ must indeed be a proof.
  \qed
\end{proof}

\Necessitation*
\begin{proof}
  By induction on $\alpha$.
  In the following, a double rule indicates zero or more applications of the indicated rule(s).
  \begin{description}[wide]
    \item[(atomic $a$):]
    We pick a fresh label $y \neq x$ and construct the following derivation.
    \begin{gather*}
      \begin{prooftree}
          \prooftree
            \prooftree
              \sequent{\Gamma}{x : \varphi}
                \using {\text{\smaller{($\square$L/WL)}}}
                \Justifies
              \sequent{y \relatedBy{a} x, \{ y : \nec{a}{\psi} \mid x : \psi \in \Gamma \}}{x : \varphi}
            \endprooftree
              \using {\text{\smaller{($\square$R)}}}
              \justifies
            \sequent{\{ y : \nec{a}{\psi} \mid x : \psi \in \Gamma \}}{y : \nec{a}{\varphi}}
          \endprooftree
          \using {\text{\smaller{(Subst)}}}
          \justifies
        \sequent{\nec{a}{\Gamma}}{x : \nec{a}{\varphi}}
      \end{prooftree}
    \end{gather*}
    In the case that $\Gamma$ is empty, we apply an instance of the (WL) rule rather than ($\square$L).
    Note here that there is a single open leaf, and no infinite paths in the derivation.
    It is clear that, for any trace value $\tau = x : \varphi$, we can form a trace $\nec{a}{\tau}, y : \nec{a}{\tau}, \tau, \ldots, \tau$ covering the path from the conclusion to the open leaf.

    \item[(composition $\comp{\alpha}{\beta}$):]
    We construct the following derivation.
    \begin{gather*}
      \begin{prooftree}
        \prooftree
          \prooftree
            \prooftree
              \sequent{\Gamma}{x : \varphi}
                \using {\hbox to -2pt {\text{\smaller{ind. hyp.}}}}
                \leadsto
              \sequent{\nec{\beta}{\Gamma}}{x : \nec{\beta}{\varphi}}
            \endprooftree
              \using {\hbox to -2pt {\text{\smaller{ind. hyp.}}}}
              \leadsto
            \sequent{\nec{\alpha}{\nec{\beta}{\Gamma}}}{x : \nec{\alpha}{\nec{\beta}{\varphi}}}
          \endprooftree
            \using {\text{\smaller{($\compSymb$L)}}}
            \Justifies
          \sequent{\nec{\comp{\alpha}{\beta}}{\Gamma}}{x : \nec{\alpha}{\nec{\beta}{\varphi}}}
        \endprooftree
          \using {\text{\smaller{($\compSymb$R)}}}
          \justifies
        \sequent{\nec{\comp{\alpha}{\beta}}{\Gamma}}{x : \nec{\comp{\alpha}{\beta}}{\varphi}}
      \end{prooftree}
    \end{gather*}
    For a trace value $\tau = x : \varphi$, take any path the from the conclusion to an open leaf and let $\vec{\tau_1}$ and $\vec{\tau_2}$ be the traces covering the portions of the path through the subderivations obtained from the induction.
    We can then form the trace $\nec{\comp{\alpha}{\beta}}{\tau}, \nec{\alpha}{\nec{\beta}{\tau}}, \ldots, \nec{\alpha}{\nec{\beta}{\tau}} \cdot \vec{\tau_1} \cdot \vec{\tau_2}$ covering the the full path.
    Any infinite paths in the derivation must be contained in one of the subderivations obtained from the induction, thus we also obtain the required infinitely progressing trace.
    
    \item[(choice $\choice{\alpha}{\beta}$):]
    We construct the following derivation.
    \begin{gather*}
      \begin{prooftree}
        \prooftree
          \prooftree
            \prooftree
              \sequent{\Gamma}{x : \varphi}
                \using {\hbox to -2pt {\text{\smaller{ind. hyp.}}}}
                \leadsto
              \sequent{\nec{\alpha}{\Gamma}}{x : \nec{\alpha}{\varphi}}
            \endprooftree
              \using {\text{\smaller{(WL)}}}
              \Justifies
            \sequent{\nec{\alpha}{\Gamma}, \nec{\beta}{\Gamma}}{x : \nec{\alpha}{\varphi}}
          \endprooftree
            \quad
          \prooftree
            \prooftree
              \sequent{\Gamma}{x : \varphi}
                \using {\hbox to -2pt {\text{\smaller{ind. hyp.}}}}
                \leadsto
              \sequent{\nec{\beta}{\Gamma}}{x : \nec{\beta}{\varphi}}
            \endprooftree
              \using {\text{\smaller{(WL)}}}
              \Justifies
            \sequent{\nec{\alpha}{\Gamma}, \nec{\beta}{\Gamma}}{x : \nec{\beta}{\varphi}}
          \endprooftree
            \using {\text{\smaller{($\choiceSymb$R)}}}
            \justifies
          \sequent{\nec{\alpha}{\Gamma}, \nec{\beta}{\Gamma}}{x : \nec{\choice{\alpha}{\beta}}{\varphi}}
        \endprooftree
          \using {\text{\smaller{($\choiceSymb$L)}}}
          \Justifies
        \sequent{\nec{\choice{\alpha}{\beta}}{\Gamma}}{x : \nec{\choice{\alpha}{\beta}}{\varphi}}
      \end{prooftree}
    \end{gather*}
    For a trace value $\tau = x : \varphi$, take any path the from the conclusion to an open leaf and let $\vec{\tau'}$ be the trace covering the portion of the path through the subderivation obtained from the induction.
    We can then form the trace $\nec{\choice{\alpha}{\beta}}{\tau}, \ldots, \nec{\choice{\alpha}{\beta}}{\tau}, \nec{\alpha}{\tau}, \ldots, \nec{\alpha}{\tau} \cdot \vec{\tau'}$, if the path traverses the left--hand premise of the instance of the ($\choiceSymb$R) rule, and $\nec{\choice{\alpha}{\beta}}{\tau}, \ldots, \nec{\choice{\alpha}{\beta}}{\tau}, \nec{\beta}{\tau}, \ldots, \nec{\beta}{\tau} \cdot \vec{\tau'}$ otherwise.
    Note that this is the only case to cause the resulting derivation to have more than one open leaf.
    Again, any infinite paths in the derivation must be contained in one of the subderivations obtained from the induction, thus we also obtain the required infinitely progressing trace.

    \item[(test $\test{\varphi'}$):]
    Without loss of generality, suppose $\Gamma = \{ x : \psi_1, \ldots, x : \psi_n \}$.
    Then define sets $\Gamma_1, \ldots, \Gamma_n$ by $\Gamma_i = \{ \nec{\test{\varphi'}}{\psi_1}, \ldots, \nec{\test{\varphi'}}{\psi_{i}}, \psi_{i+1}, \ldots, \psi_n \}$.
    We construct a series of open derivations $\mathcal{D}_1, \ldots, \mathcal{D}_n$ inductively as follows.
    \begin{gather*}
      \mathcal{D}_1 =
      \begin{gathered}
        \begin{prooftree}
          \prooftree
            \prooftree
              \prooftree
                  \using {\text{\smaller{(Ax)}}}
                  \justifies
                \sequent{x : \varphi'}{x : \varphi'}
              \endprooftree
                \using {\text{\smaller{(WR)}}}
                \justifies
              \sequent{x : \varphi'}{x : \varphi, x : \varphi'}
            \endprooftree
              \using {\text{\smaller{(WL)}}}
              \Justifies
            \sequent{x : \varphi', \Gamma_1}{x : \varphi, x : \varphi'}
          \endprooftree
            \quad
          \prooftree
            \sequent{\Gamma}{x : \varphi}
              \using {\text{\smaller{(WL)}}}
              \justifies
            \sequent{x : \varphi', \Gamma}{x : \varphi}
          \endprooftree
            \using {\text{\smaller{($\testSymb$L)}}}
            \justifies
          \sequent{x : \varphi', \Gamma_1}{x : \varphi}
        \end{prooftree}
      \end{gathered}
    \end{gather*}
    \begin{gather*}
      \mathcal{D}_{i+1} =
      \begin{gathered}
        \begin{prooftree}
          \prooftree
            \prooftree
              \prooftree
                  \using {\text{\smaller{(Ax)}}}
                  \justifies
                \sequent{x : \varphi'}{x : \varphi'}
              \endprooftree
                \using {\text{\smaller{(WR)}}}
                \justifies
              \sequent{x : \varphi'}{x : \varphi, x : \varphi'}
            \endprooftree
              \using {\text{\smaller{(WL)}}}
              \Justifies
            \sequent{x : \varphi', \Gamma_{i+1}}{x : \varphi, x : \varphi'}
          \endprooftree
            \quad
          \prooftree
            \mathcal{D}_{i}
              \leadsto
            \sequent{x : \varphi', \Gamma_i}{x : \varphi}
          \endprooftree
            \using {\text{\smaller{($\testSymb$L)}}}
            \justifies
          \sequent{x : \varphi', \Gamma_{i+1}}{x : \varphi}
        \end{prooftree}
      \end{gathered}
    \end{gather*}
    The required derivation is the following.
    \begin{gather*}
      \begin{prooftree}
        \prooftree
          \mathcal{D}_{n}
            \leadsto
          \sequent{x : \varphi', \Gamma_{n}}{x : \varphi}
        \endprooftree
          \using {\text{\smaller{($\testSymb$R)}}}
          \justifies
        \sequent{\nec{\test{\varphi'}}{\Gamma}}{x : \nec{\test{\varphi'}}{\varphi}}
      \end{prooftree}
    \end{gather*}
    If $\Gamma$ is empty, then we may just apply the (WL) rule instead of using derivation $\mathcal{D}_n$.
    Notice that there is a single open leaf in this derivation and, thus a single path from the conclusion to this leaf that traverses the right-hand premise of each instance of the ($\testSymb$L) rule.
    For a trace value $\tau = x : \varphi$, w can construct the trace $\nec{\test{\varphi'}}{\tau}, \tau, \ldots, \tau$ that covers this path.
    Notice also that there are no infinite paths in this derivation.

    \item[(iteration $\iter{\alpha}$):]
    We construct the following derivation.
    \begin{gather*}
      \begin{prooftree}
        \prooftree
          \prooftree
            \sequent{\Gamma}{x : \varphi}
              \using {\text{\smaller{(WL)}}}
              \Justifies
            \sequent{\Gamma, \nec{\alpha}{\nec{\iter{\alpha}}{\Gamma}}}{x : \varphi}
          \endprooftree
            \quad
          \prooftree
            \prooftree
                \begin{gathered}
                  \tikz \coordinate (bud) at (0,-0.5em);
                    \\
                  \sequent{\nec{\iter{\alpha}}{\Gamma}}{x : \nec{\iter{{\color{blue}{\underline{\alpha}}}}}{\varphi}}
                \end{gathered}
                \using {\hbox to -2pt {\text{\smaller{ind. hyp.}}}}
                \leadsto
              \sequent{\nec{\alpha}{\nec{\iter{\alpha}}{\Gamma}}}{x : \nec{\alpha}{\nec{\iter{{\color{blue}{\underline{\alpha}}}}}{\varphi}}}
            \endprooftree
              \using {\text{\smaller{(WL)}}}
              \Justifies
            \sequent{\Gamma, \nec{\alpha}{\nec{\iter{\alpha}}{\Gamma}}}{x : \nec{\alpha}{\nec{\iter{{\color{blue}{\underline{\alpha}}}}}{\varphi}}}
          \endprooftree
            \using {\text{\smaller{($\ast$R)}} \; \dag}
            \justifies
          \sequent{\Gamma, \nec{\alpha}{\nec{\iter{\alpha}}{\Gamma}}}{x : \nec{\iter{{\color{blue}{\underline{\alpha}}}}}{\varphi}}
        \endprooftree
          \using {\text{\smaller{($\ast$L)}}}
          \Justifies
        \sequent{\nec{\iter{\alpha}}{\Gamma}}{x : \nec{\iter{{\color{blue}{\underline{\alpha}}}}}{\varphi}}
        \tikz \coordinate (companion) at (0.5em,0.25em);
      \end{prooftree}
      \begin{tikzpicture}[overlay]
        \draw[arrow,dashed] (bud) -- ++(0,1em) -- ++(9em,0) |- (companion) [->];
      \end{tikzpicture}
    \end{gather*}
    Here, we have converted each open leaf in the subderivation obtained via the inductive hypothesis into buds, with associated companion the sequent concluding the full derivation.
    This derivation thus has a single open leaf, above the left-hand premise of the instance of the ($\ast$R) rule.
    The inductive hypothesis guarantees that every path from the conclusion of the subderivation to a bud is covered by a trace $\nec{\alpha}{\nec{\iter{\alpha}}{\tau}}, \ldots, \nec{\iter{\alpha}}{\tau}$ for each trace value $\tau = x : \varphi$.
    Thus, for each such trace value, there is a trace $\nec{\iter{\alpha}}{\tau}, \ldots, \nec{\iter{\alpha}}{\tau}$ covering each cycle from the conclusion of the derivation back to itself.
    Now, any path from the conclusion to the open leaf must traverse a (finite) sequence of such cycles, before then traversing the left-hand premise of the instance of the ($\ast$R) rule marked by $\dag$ and continuing to the leaf.
    We can thus concatenate the traces covering each basic cycle, and preprend this to the trace $\nec{\iter{\alpha}}{\tau}, \ldots, \nec{\iter{\alpha}}{\tau}, \tau, \ldots, \tau$ covering the final portion of the path from the conclusion to the open leaf.
    This shows that \cref{cond:paths:finite} holds.

    To see that \cref{cond:paths:infinite} holds, notice that the infinite paths in this derivation fall into two categories: either the path remains within the subderivation obtained via the induction; or the path visits the conclusion of the derivation infinitely often.
    In the former case, we obtain by the induction that the path is followed by an infinitely progressing trace.
    In the latter, it must consist of an infinite sequence of cycles from the conclusion.
    Each such cycle is covered by a trace, as indicated by the underlined programs highlighted in blue, which progresses at the instance of the ($\ast$) rule marked with $\dag$.
    Thus we have an infinitely progressing trace by concatenating these individual traces.
    \qed

  \end{description}
\end{proof}

\RegularCompleteness*
\begin{proof}
  If $\models \varphi$ then there is a Hilbert deduction $\psi_1, \ldots, \psi_n$ where $\varphi = \psi_n$.
  From such a deduction, we define a sequence $\mathcal{D}_1, \ldots, \mathcal{D}_n$ of {\CyclicSystem} proofs inductively as follows.
  For each $\psi_i$, if it is an instance of an axiom, then $\mathcal{D}_i$ is the corresponding derivation of $\sequent{}{x : \psi_i}$.
  If $\psi_i$ is derived from $\psi_j$ ($j < i$) via \eqref{ax:PDL:Necessitation}, then $\mathcal{D}_i$ is the derivation constructed by applying \cref{lem:Necessitation} for the sequent $\sequent{}{x : \psi_j}$ and replacing each open leaf with a copy of $\mathcal{D}_j$.
  If $\psi_i$ is derived from $\psi_j$ and $\psi_k$ ($j, k < i$) via \eqref{ax:ModusPonens}, then $\mathcal{D}_i$ is the following derivation.
  \begin{gather*}
    \resizebox{\textwidth}{!}{
    \begin{prooftree}
        \prooftree
          \prooftree
            \mathcal{D}_j
              \leadsto
            \sequent{}{x : \psi_j}
          \endprooftree
            \quad
          \prooftree
            \mathcal{D}_k
              \leadsto
            \sequent{}{x : \psi_j \rightarrow \psi_i}
          \endprooftree
            \using {\text{\smaller{($\wedge$R)}}}
            \justifies
          \sequent{}{x : \psi_j \wedge (\psi_j \rightarrow \psi_i)}
        \endprooftree
          \;
        \prooftree
          \prooftree
            \prooftree
              \prooftree
                  \using {\text{\smaller{(Ax)}}}
                  \justifies
                \sequent{x : \psi_j}{x : \psi_j}
              \endprooftree
                \using{\text{\smaller{(WR)}}}
                \justifies
              \sequent{x : \psi_j}{x : \psi_i, x : \psi_j}
            \endprooftree
              \;
            \prooftree
              \prooftree
                  \using {\text{\smaller{(Ax)}}}
                  \justifies
                \sequent{x : \psi_i}{x : \psi_i}
              \endprooftree
                \using{\text{\smaller{(WL)}}}
                \justifies
              \sequent{x : \psi_j, x : \psi_i}{x : \psi_i}
          \endprooftree
            \using {\text{\smaller{($\rightarrow$L)}}}
              \justifies
            \sequent{x : \psi_j, x : (\psi_j \rightarrow \psi_i)}{x : \psi_i}
          \endprooftree
            \using {\text{\smaller{($\wedge$L)}}}
            \justifies
          \sequent{x : \psi_j \wedge (\psi_j \rightarrow \psi_i)}{x : \psi_i}
        \endprooftree
        \using {\text{\smaller{(Cut)}}}
        \justifies
      \sequent{}{x : \psi_i}
    \end{prooftree}}
  \end{gather*}
  Thus, $\mathcal{D}_n$ is a {\CyclicSystem} proof of $\sequent{}{x : \varphi}$.
  \qed
\end{proof}

\subsection*{Proof Search in {\CyclicSystem} for Test-free, Acyclic Sequents}

\ValidWeakening*
\begin{proof}
  \begin{enumerate}[itemsep={\medskipamount}]
    \item
    If $\sequent{\Gamma}{\Delta, x \relatedBy{a} z}$ is valid then it has a cut-free (and substitution-free) {\InfinitarySystem} proof $\mathcal{P}$, by \cref{thm:Cut-freeCompleteness}.
    Since $x \relatedBy{a} z \not\in \Gamma$, we must also have $x \relatedBy{a} z \not\in \Gamma'$ for any sequent $\sequent{\Gamma'}{\Delta'}$ in the proof.
    This follows from the fact that if $x \relatedBy{a} z$ is in the consequent but not the antecedent of the conclusion of a logical inference rule, then this property also holds of its premises.
    This is immediate for all rules except ($\square$R), which do not introduce new relational atoms in the premises and preserve their contexts.
    For ($\square$R), notice that although it may introduce a relational atom $x \relatedBy{a} y$ in the premise we must have $y \neq z$ since $x \relatedBy{a} z$ appears in the conclusion and $y$ must be fresh.
    Note that although weakening rules may be used to remove $x \relatedBy{a} z$ from the consequent, these are only used in the proof to lead to an axiom and so $x \relatedBy{a} z$ cannot occur in the antecedent of any sequent above such a weakening rule.
    The relational atom $x \relatedBy{a} z$ is thus never principal for any rule in $\mathcal{P}$, and so we can produce a {\InfinitarySystem} proof of $\sequent{\Gamma}{\Delta}$ by removing $x \relatedBy{a} z$ from every consequent in $\mathcal{P}$ in which it appears.
    Then, by \cref{thm:Soundness:Inf}, $\sequent{\Gamma}{\Delta}$ is valid.

    \item
    Since $\sequent{\Gamma}{\Delta, x : p}$ is valid it must have a cut-free {\InfinitarySystem} proof $\mathcal{P}$, by \cref{thm:Cut-freeCompleteness}.
    We show that every sequent $\sequent{\Gamma'}{\Delta'}$ in $\mathcal{P}$ satisfies: 
    \begin{enumerate*}[label={(\roman*)},ref={\roman*}]
      \item
      \label[invariant]{inv:Right:Formula:NoAxiom}
      $x : p \not\in \Gamma'$;
      \item
      \label[invariant]{inv:Right:Formula:Reachability:Right}
      $z \neq x$ and $z$ does not reach $x$ in $\Gamma'$ for all $z : \varphi \in \Delta'$ with $\varphi$ non-atomic;
      \item
      \label[invariant]{inv:Right:Formula:NoCompositesOrIterates}
      $z \neq x$ for all $z : \varphi \in \Gamma'$ with $\varphi$ composite or iterated;
      \item
      \label[invariant]{inv:Right:Formula:NoActivatingRelationalAtoms}
      if $x \relatedBy{a} z \in \Gamma'$ then there is no non-atomic $\varphi$ such that $x : \varphi \in \Gamma'$; and
      \item
      \label[invariant]{inv:Right:Formula:Reachability:Left}
      if $z : \varphi \in \Gamma'$ with $\varphi$ non-atomic then $z$ does not reach $x$ in $\Gamma'$.
    \end{enumerate*}
    Note that the root sequent $\sequent{\Gamma}{\Delta, x : p}$ satisfies these since it is normal and $x \not\in \starredlabs{\Delta}$.
    Furthermore, the proof rules used in $\mathcal{P}$ preserve these properties from conclusion to premise.
    We take each invariant in turn.
    \begin{description}[font={\normalfont}]
      \item[(\ref{inv:Right:Formula:NoAxiom}):]
      For rules ($\wedge$L), ($\vee$L), ($\rightarrow$L) and ($\ast$L), this follows from the fact that the conclusion satisfies \cref{inv:Left:Formula:NoCompositesOrIterates}.
      For rule ($\rightarrow$R) and ($\testSymb$R), this follows from the fact that the conclusion satisfies \cref{inv:Right:Formula:Reachability:Right}.
      The result is straightforward or immediate for the other rules.
      \item[(\ref{inv:Right:Formula:Reachability:Right}):]
      For ($\square$R), it is immediate since $y$ is fresh.
      For rules ($\rightarrow$L) and ($\testSymb$L), this follows from the fact that the conclusion satisfies \cref{inv:Right:Formula:NoCompositesOrIterates}.
      The result is straightforward or immediate for the other rules.
      \item[(\ref{inv:Right:Formula:NoCompositesOrIterates}):]
      For rules ($\rightarrow$R) and ($\testSymb$R), this follows from the fact that the conclusion satisfies \cref{inv:Right:Formula:Reachability:Right}.
      For rule ($\square$R), this follows from the fact that the conclusion satisfies \cref{inv:Right:Formula:NoActivatingRelationalAtoms}.
      The result is straightforward or immediate for the other rules.
      \item[(\ref{inv:Right:Formula:NoActivatingRelationalAtoms}):]
      For rules ($\rightarrow$R) and ($\testSymb$R), this is immediate since from the fact that the conclusion satisfies \cref{inv:Left:Formula:NoCompositesOrIterates} we have $z \neq x$ for the principal formula $x : \varphi$.
      For rule ($\square$L), we must have that $x \neq z$ for the principal formula $z : \nec{a}{\varphi}$; but then because the conclusion satisfies \cref{inv:Right:Formula:Reachability:Left} we have that $z$ does not reach $x$ in $\Gamma'$ and so $y \neq x$.
      For ($\square$R), we have that $z \neq x$ for the principal formula $z : \nec{a}{\varphi}$ since the conclusion satisfies \cref{inv:Left:Formula:NoCompositesOrIterates}; so the result follows trivially since we the new relational atom in the premise does not involve $x$.
      The result is straightforward or immediate for the other rules.
      \item[(\ref{inv:Right:Formula:Reachability:Left}):]
      For rule ($\square$L), the result is immediate since $y$ in the premise is fresh.
      For rule ($\square$R), we have $x \neq z$ for the principal formula $z : \nec{a}{\varphi}$ since the conclusion satisfies \cref{inv:Left:Formula:NoActivatingRelationalAtoms}, and by assumption that $z$ does not reach $x$ in $\Gamma'$; but then $y$ does not reach $x$ in $\Gamma'$, otherwise then so would $z$.
      For rules ($\rightarrow$) and ($\testSymb$R), the result is immediate since the conclusion satisfies \cref{inv:Right:Formula:Reachability:Right}.
      For the remaining right-hand rules the result is immediate since they do not introduce new labelled formulas in the antecedents of the premises.
      For the remaining left-hand rules, it is also straightfoward since there is $z : \varphi$ with $\varphi$ non-atomic in the antecedent of a premise only if there is some $z : \psi$ with $\psi$ non-atomic in the antecedent of the conclusion.
    \end{description}
    Therefore, the formula $x : p$ is never principal for any rule in $\mathcal{P}$, and so we can produce a {\InfinitarySystem} proof of $\sequent{\Gamma}{\Delta}$ by removing $x : p$ from every consequent in $\mathcal{P}$ in which it appears.
    Then, by \cref{thm:Soundness:Inf}, $\sequent{\Gamma}{\Delta}$ is valid.

    \item
    Since $\sequent{\Gamma, x : \varphi}{\Delta}$ is valid it must have a cut-free {\InfinitarySystem} proof $\mathcal{P}$, by \cref{thm:Cut-freeCompleteness}.
    We show that every sequent $\sequent{\Gamma'}{\Delta'}$ in $\mathcal{P}$ satisfies: 
    \begin{enumerate*}[label={(\roman*)},ref={\roman*}]
      \item
      \label[invariant]{inv:Left:Formula:NotInConsequentLabels}
      $x \not\in \labs{\Delta'}$; 
      \item
      \label[invariant]{inv:Left:Formula:NoCompositesOrIterates}
      $z \neq x$ for all $z : \varphi \in \Gamma'$ with $\varphi$ composite or iterated; and
      \item
      \label[invariant]{inv:Left:Formula:NoActivatingRelationalAtoms}
      for any program $a$ and label $z$, it holds that $x \relatedBy{a} z \not\in \Gamma'$.
    \end{enumerate*}
    The root sequent $\sequent{\Gamma, x : \varphi}{\Delta}$ satisfies these since it is normal (in particular, $\Gamma$ contains no labelled composite or iterated formulas), and $x \not\in \labs{\Delta}$.
    Furthermore, the proof rules used in $\mathcal{P}$ preserve these properties from conclusion to premise.
    The non-trivial cases are the rules for atomic modalities and implications.
    In ($\square$L), since there is no relational atom $x \relatedBy{a} y$ in the antecedent of the conclusion, the active formula $z : \nec{a}{\psi}$ must satisfy $z \neq x$, and so \cref{inv:Left:Formula:NoCompositesOrIterates} is maintained in the premise.
    In ($\square$R), since $x$ is not in the labels of the consequent of the conclusion, the new relational atom in the premise preserves \cref{inv:Left:Formula:NoActivatingRelationalAtoms}; and because the new label in the premise is fresh, \cref{inv:Left:Formula:NotInConsequentLabels} is also maintained.
    In ($\rightarrow$L), since the active formula $z : \psi \rightarrow \psi'$ in the conclusion satisfies $z \neq x$, the new formula $z : \psi$ in the consequent of the left-hand premise maintains \cref{inv:Left:Formula:NotInConsequentLabels}.
    Similarly for ($\testSymb$L).
    In ($\rightarrow$R) and ($\testSymb$R), since $x$ does not occur in the labels of the consequent of the conclusion, the new formula in the antecedent of the premise maintains \cref{inv:Left:Formula:NoCompositesOrIterates}.
    Thus, the formula $x : \varphi$ is never for any rule in $\mathcal{P}$, and so we can produce a {\InfinitarySystem} proof of $\sequent{\Gamma}{\Delta}$ by removing $x : \varphi$ from every antecedent in $\mathcal{P}$ in which it appears.
    Then, by \cref{thm:Soundness:Inf}, $\sequent{\Gamma}{\Delta}$ is valid.

    \item
    Since $\sequent{\Gamma, x \relatedBy{a} y}{\Delta}$ is valid, by the construction for \cref{thm:Cut-freeCompleteness} it must have a cut-free {\InfinitarySystem} proof $\mathcal{P}$ such that every sequent $\sequent{\Gamma'}{\Delta'}$ in $\mathcal{P}$ (apart from the root) that is not the premise of a weakening rule contains its conclusion (i.e.~its antecedent is a superset of the antecedent of its conclusion, and similarly for its consequent).
    It is straightforward to see that, for every rule, when the conclusion $\sequent{\Gamma'}{\Delta'}$ satisfies the property that $z : \varphi \in \Gamma'$ or $z : \varphi \in \Delta'$ implies either $z \in \labs{\Delta}$ or there is some $z' \in \labs{\Delta}$ that reaches $z$ in $\Gamma'$, then so does the premise.
    The only non-trivial cases are the atomic modality rules.
    For ($\square$R) rule, we have from the assumption on the conclusion that either $z \in \labs{\Delta}$ or there is $z' \in \labs{\Delta}$ that reaches $z$ in $\Gamma'$ for the principal formula $z : \nec{a}{\varphi}$.
    Therefore, because the relational atom $z \relatedBy{a} y$ is introduced in the premise, we have that $z$ or $z'$ reaches $y$ in the premise antecedent.
    The ($\square$L) rule is, in fact, unproblematic since we have that each rule in $\mathcal{P}$ preserves the active formula(s) of the conclusion in the context of the premises; thus, the reachability relation induced by the antecedents only grows from conclusion to premise.
    Thus, every sequent in $\mathcal{P}$ satisfies this property, whence the we derive that the property holds for the premise.
    Now, it follows that the relational atom $x \relatedBy{a} y$ can never be active in $\mathcal{P}$; if it were, there would be some sequent $\sequent{\Gamma'}{\Delta'}$ such that some corresponding labelled formula $x : \nec{a}{\varphi}$ is in $\Gamma'$.
    But then we would have that either $x \in \labs{\Delta}$ or $x$ is reachable from some label in $\Delta$, which contradicts the assumption that $x \not\in \labs{\Delta}$ nor reachable from any label in $\Delta$.
    Therefore we can produce a {\InfinitarySystem} proof of $\sequent{\Gamma}{\Delta}$ by removing $x \relatedBy{a} y$ from every antecedent in $\mathcal{P}$ in which it appears.
    Then, by \cref{thm:Soundness:Inf}, $\sequent{\Gamma}{\Delta}$ is valid.
    \qed
  \end{enumerate}
\end{proof}

\UnwindingProperties*
\begin{proof}
  \begin{enumerate}
    \item
    $\Gamma' \cap \Delta' = \emptyset$ otherwise it would be closed by weakening and an instance of (Ax).
    $\Delta'$ can only contain labelled atomic and iterated formulas that have already been unfolded along the path from the root of the unwinding, otherwise the construction of the unwinding would be able to apply further logical rules.
    The third condition of \cref{def:NormalSequents} must hold similarly: otherwise the unwinding construction could apply further logical rules to $\Gamma'$.
    Moreover, the weakening steps of \cref{lem:ValidWeakening} preserve these properties, since they only remove elements from sequents, and remove all relational atoms from the consequent.
    \item
    In uncapped unwindings only logical rules are applied.
    Notice that these are invertible.
    \Cref{lem:ValidWeakening} shows that the sequence of weakening rules applied to the open leaves of unwindings preserve validity.
    \item
    Notice that all rules except ($\ast$L) and ($\ast$R) have the sub-formula property.
    Thus, if an ancestor formula contains its descendant as a sub-formula then the path of ancestry must be principal for one of these unfolding rules, whence we derive the result.
    \qed
  \end{enumerate}
\end{proof}

To prove \cref{lem:FiniteUnwindings}, we need the following three results.
\Cref{lem:UnwindingsPreserveAcyclicity} entails that the proof-search procedure preserves acyclicity.
We use \cref{lem:TestfreeUnwindings:FiniteLeftUnfoldings} to show that in unwindings we cannot keep unfolding iterated formulas on the left forever.
\Cref{lem:TestfreeUnwindings:TerminalLabels} gives that unwindings produce consequent formulas with terminal labels, which we use to show that open leaves of capped unwindings do not contain relational atoms.

\begin{proposition}
\label{lem:UnwindingsPreserveAcyclicity}
  Let $\mathcal{D}$ be a finite derivation consisting only of logical rules concluding with an acyclic sequent; then every sequent appearing in $\mathcal{D}$ is acyclic.
\end{proposition}
\begin{proof}
  By induction on the structure of derivations.
  The only non-trivial cases are the rules for atomic modalities.
  For ($\square$L), assume the conclusion is acyclic.
  Then the premise must be acyclic since its set of antecedent relational atoms is a subset of those of the conclusion.
  Then the result follows by induction.
  For ($\square$R), assume the conclusion is acyclic.
  Then the premise is acyclic since the new relational atom $x \relatedBy{a} y$ does not link to any existing label, $y$ being fresh.
  Then the result follows by induction.
\end{proof}

\begin{proposition}
\label{lem:TestfreeUnwindings:FiniteLeftUnfoldings}
  Let $\mathcal{D}$ be a finite derivation consisting only of left logical rules and concluding with a test-free sequent $\sequent{\Gamma}{\Delta}$; then if it contains a path of ancestry from some ancestor antecedent formula $y : \varphi$ to a descendent antecedent formula $x : \psi \in \Gamma$, with $\varphi$ a sub-formula of $\psi$, then $x$ reaches $y$ in $\Gamma$.
\end{proposition}
\begin{proof}
  By straightforward induction on the structure of derivations.
\end{proof}

\begin{proposition}
\label{lem:TestfreeUnwindings:TerminalLabels}
  Let $\mathcal{D}$ be a possibly open finite derivation consisting only of logical rules that preserve relational atoms from conclusion to premise, and let $\sequent{\Gamma}{\Delta}$ be an open leaf of $\mathcal{D}$; if there is a path of ancestry from $y : \varphi \in \Delta$ that is principal for an instance of {\normalfont{($\square$R)}}, then $y$ does not reach any label $z$ in $\Gamma'$.
\end{proposition}
\begin{proof}
  By induction on the structure of derivations.
  Most cases are straightforward.
  The only interesting cases are those for the atomic modalities.
  For ($\square$L), assume the result holds for the premise; then it must continue to hold for the conclusion because the principal relational atom must already be in the antecedent of the premise, since the derivation only uses rules that preserve relational atoms from conclusion to premise.
  Thus, we do not increase reachability.
  For ($\square$R), the result holds by induction for formulas in the context, and if there is a path of ancestry already principal for an instance of ($\square$R) from $y : \varphi \in \Delta$ to the principal formula in the premise.
  However here we need to further consider the following.
  Assume that there is a path of ancestry from $y : \varphi \in \Delta$ to the principal formula $z : \psi$ in the premise with $\varphi$ a sub-formula of $\psi$, but not principal for an instance of ($\square$R).
  Thus, we must have $z = y$.
  The immediate descendant of this formula, $x : \nec{a}{\psi}$, is now connected by a path of ancestry that is principal for ($\square$R) to $y : \varphi \in \Delta$.
  So we must show that $y$ does not reach any label $x$ in $\Gamma$.
  Note that $y$ must be fresh for the contexts $\Gamma'$ and $\Delta'$ of the conclusion.
  Now, suppose $y$ reaches some label $z'$ in $\Gamma$; then there must be some relational atom $y \relatedBy{b} z'$ introduced into the left-hand context along the path up to the open leaf $\sequent{\Gamma}{\Delta}$.
  But since the path of ancestry leading to $y : \varphi$ is never principal for ($\square$R), this can only happen if there is an instance of ($\square$R) along the path that is principal for some other formula $y : \nec{b}{\varphi'}$.
  However if this were the case then it would have a descendant in $\Gamma'$ that is \emph{distinct} from $y : \psi$.
  This would contradict the side-condition of the rule: that $y$ is fresh, so we conclude that $y$ does not reach any label $x$ in $\Gamma$ after all.
\end{proof}

We now give the proof of \cref{lem:FiniteUnwindings}.

\FiniteUnwindings*
\begin{proof}
  For finiteness, it suffices to show that uncapped unwindings for test-free, acyclic sequents are finite, since no logical rules are applied during the capping process and the weakening rules strictly reduce the size of the sequent they are applied to.
  Suppose, for contradiction, that the unwinding is infinite.
  Then there must be an infinite path in the derivation containing either an infinite number of right-hand rules or an infinite number of left-hand rules.
  If it is the former then, since each rule in the derivation has the sub-formula property and each formula that appears is in the (finite) Fischer-Ladner closure of $\sequent{\Gamma}{\Delta}$, there must be a consequent formula that is principal for ($\ast$R) an infinite number of times.
  However, this is impossible since each consequent formula in an unwinding can only be unfolded at most once.
  So we may assume there are only a finite number of right-hand rules.
  If, then, there are an infinite number of left-hand rules, we must similarly have that there is a path of ancestry in which some antecedent formula $\varphi$ appears an infinite number of times.
  Since there are only a finite number of right-hand rules, there is an infinite sub-derivation $\mathcal{D'}$ consisting only of left logical rules, and concluding with some sequent $\sequent{\Gamma'}{\Delta'}$.
  Now, take a (necessarily finite) prefix of $\mathcal{D'}$ such that there are $n$ occurrences of the ($\ast$L) rule for which the path of ancestry is principal, $n$ being the number of distinct labels in $\Gamma'$.
  By \cref{lem:TestfreeUnwindings:FiniteLeftUnfoldings}, we must have that there are $n+1$ labels $x_1, \ldots, x_{n+1}$ such that $x_i$ reaches each $x_j$ in $\Gamma'$ ($j > i$), where $x_1 : \varphi \in \Gamma'$ is on the path of ancestry.
  Therefore, we must have $x_i = x_j$ for some $i$ and $j$ with $1 \leq i < j \leq n+1$.
  That is, $\Gamma'$ is cyclic.
  However this is also impossible by \cref{lem:UnwindingsPreserveAcyclicity}, since we have assumed that $\sequent{\Gamma}{\Delta}$ is acyclic.
  Thus, we conclude $\mathcal{D}$ must be finite after all.
  
  The second part, that $\labs{\Gamma'} \subseteq \labs{\Delta'}$ for open leaves $\sequent{\Gamma'}{\Delta'}$, is ensured by the weakening rules that are applied to the open leaves of uncapped unwindings.
  To see that this is the case, note the following.
  The second category of weakening rules ensure that $\labs{\Delta'} \subseteq \starredlabs{\Delta'}$.
  The third category of weakening rules ensure that $x : \varphi \in \Gamma'$ only if $x \in \labs{\Delta'}$.
  The fourth category of weakening rules ensures there are no relational atoms in $\Gamma'$.
  The weakening steps remove relational atoms of the form $x \relatedBy{a} y$ such that $x \not\in \labs{\Delta'}$ and $x$ is not reachable in the antecedent of the open leaf of the unwinding from any label in $\Delta'$.
  So we are left to account for relational atoms $x \relatedBy{a} y$ such that $x \in \labs{\Delta'}$ or $x$ is reachable from some label in $\Delta'$.
  However, no such relational atoms exist.
  To see that this is the case, notice that the construction of the unwinding ensures that each $z : \varphi \in \Delta'$ with $\varphi$ iterated has been unfolded once within the unwinding and so, because $\varphi$ is test-free, this means that there is a path of ancestry from $z : \varphi$ that is principal for a ($\square$R) rule.
  Then we have from \cref{lem:TestfreeUnwindings:TerminalLabels} that any such $z$ does not reach any label in $\Gamma'$.
  Now consider the two cases.
  If $x \in \labs{\Delta'}$, then $x$ reaches $y$ since $x \relatedBy{a} y$ is in the open leaf of the unwinding; but this is impossible as we have shown that $x$ does not reach any label.
  The other case is also impossible from our previous reasoning that $x$ is not reachable from any label in $\Delta'$. 
\end{proof}

We now define the functions that provide a bound on the possible number of distinct sequents (up to label substitution) encountered during proof search for test-free, acyclic sequents.
These functions are defined for test-free formulas only.
Addition and multiplication on natural numbers are denoted $+$ and $\times$, respectively.
We use $\oplus$ to denote the $\mathsf{max}$ function on natural numbers.
The set of words in the language of (a regular expression) $\alpha$ with length at most $n$ is denoted by $\truncateAt{n}{\languageOf(\alpha)}$.

\begin{definition}[{$\pathMax$}]
\label{def:PathMax}
  We define the functions $\pathMax^{+}$ and $\pathMax^{-}$ by mutual induction as follows.
  \begin{align*}
    \pathMax^{+}(p) &= 0
      \\
    \pathMax^{+}(\varphi \wedge \psi) &= \pathMax^{+}(\varphi \vee \psi)
      = \pathMax^{+}(\varphi) \oplus \pathMax^{+}(\psi)
      \\
    \pathMax^{+}(\varphi \rightarrow \psi) &= \pathMax^{-}(\varphi) \oplus \pathMax^{+}(\psi)
      \\
    \pathMax^{+}(\nec{\alpha}{\varphi}) &= \pathMax^{+}(\varphi) \oplus \unfold(\alpha)
      \\[1em]
    \pathMax^{-}(p) &= 0
      \\
    \pathMax^{-}(\varphi \wedge \psi) &= \pathMax^{-}(\varphi \vee \psi)
      = \pathMax^{-}(\varphi) \oplus \pathMax^{-}(\psi)
      \\
    \pathMax^{-}(\varphi \rightarrow \psi) &= \pathMax^{+}(\varphi) \oplus \pathMax^{-}(\psi)
      \\
    \pathMax^{-}(\nec{\alpha}{\varphi}) &= \pathMax^{-}(\varphi)
  \end{align*}
  The function $\unfold$ on programs is defined as follows.
  \begin{align*}
    \unfold(a) &= 1
      &
    \unfold(\comp{\alpha}{\beta}) &= \unfold(\alpha) + \unfold(\beta)
      \\
    \unfold(\iter{\alpha}) &= \unfold(\alpha)
      &
    \unfold(\choice{\alpha}{\beta}) &= \unfold(\alpha) \oplus \unfold(\beta)
  \end{align*}
  We then define the $\pathMax$ function on pairs of labels and sequents, as follows.
  \begin{equation*}
    \pathMax(x, \sequent{\Gamma}{\Delta}) 
      = \bigoplus 
          \left(\begin{gathered}
            \{ \pathMax^{-}(\varphi) \mid x : \varphi \in \Gamma \} \hspace{2em}
              \\
            {} \cup \{ \pathMax^{+}(\varphi) \mid x : \varphi \in \Delta \}
          \end{gathered}\right)
  \end{equation*}
\end{definition}

Intuitively, $\pathMax$ measures the maximum length of a path of reachability from the label $x$ that will be generated, by instance of the ($\square$R) rule, in an unwinding of the sequent $\sequent{\Gamma}{\Delta}$.
The auxiliary function $\unfold$ gives the maximum length of a single `unfolding' of the program $\alpha$.

The $\pathMax^{-}$ function is used for antecedent formulas, and $\pathMax^{+}$ for consequent formulas.
For example, the case for conjunctions and disjunctions take the maximum of the values of each conjunct/disjunct because the proof rules split the formula.
The maximum path length will the maximum of the maximum lengths for each of the constituent immediate sub-formulas.
The rules for an implication $\varphi \rightarrow \psi$ switch the polarity for the antecedent sub-formula $\varphi$ because the proof rules move this sub-formula to the opposite side.

The case for iteration differs for the two polarities.
Modalities on the left do not contribute to reachability paths: the ($\square$L) is only activated by relational atoms, it does not generate them during proof-search.
Thus, the case for modalities in $\pathMax$ simply recurses on the immediate sub-formula.
Modalities on the right-hand side \emph{do} contribute to reachability.
The ($\ast$R) rule has two premises: one simply strips off the modality, whereas the other one unfolds it.
Therefore, the maximum path of reachability generated by an iterated formula is the maximum of that of its immediate sub-formula, on the one hand, and the maximum length of one copy of the iterated program $\alpha$, on the other.
This is because an unwinding will unfold an iterated formula exactly once.

The function $\unfold$ gives the maximum length of one copy of a program $\alpha$.
Atomic programs have unit length; a composition has the sum of the lengths of the sub-programs, while with choice we take the maximum.
As mentioned, iterations are unfolded exactly one, so for $\iter{\alpha}$, $\unfold$ simply recurses on the sub-program $\alpha$.

The maximum length of the path of reachability generated by an unwinding from a label $x$ in the initial sequent $\sequent{\Gamma}{\Delta}$ is then simply the maximum of those of the formulas $\varphi$ that appear in $\sequent{\Gamma}{\Delta}$ labelled with $x$.

We now define the function $\starMax$, which, intuitively, gives an upper bound on the number of number of labelled iterated formulas that will appear in the consequent (right-hand side) of the sequent in the open leaves of an unwinding of $\sequent{\Gamma}{\Delta}$.

\begin{definition}[{$\starMax$}]
\label{def:StarMax}
  We define the functions $\starMax^{+}$ and $\starMax^{-}$, which operate on pairs of natural numbers and (test-free) formulas, by mutual induction as follows.
  \begin{align*}
    \starMax^{+}(n, p) &= 0
      \\
    \starMax^{+}(n, \varphi \wedge \psi) &= \starMax^{+}(n, \varphi) \oplus \starMax^{+}(n, \psi)
      \\
    \starMax^{+}(n, \varphi \vee \psi) &= \starMax^{+}(n, \varphi) + \starMax^{+}(n, \psi)
      \\
    \starMax^{+}(n, \varphi \rightarrow \psi) &= \starMax^{-}(n, \varphi) + \starMax^{+}(n, \psi)
      \\
    \starMax^{+}(n, \nec{\alpha}{\varphi}) &=
      \left\{
        \begin{aligned}
          & \starMax^{+}(n, \varphi) && \text{if $\alpha$ is $\ast$-free}
            \\
          & 1 \oplus \starMax^{+}(n, \varphi) && \text{otherwise}
        \end{aligned}
      \right.
      \\[1em]
    \starMax^{-}(n, p) &= 0
      \\
    \starMax^{-}(n, \varphi \wedge \psi) &= \starMax^{-}(n, \varphi) + \starMax^{-}(n, \psi)
      \\
    \starMax^{-}(n, \varphi \vee \psi) &= \starMax^{-}(n, \varphi) \oplus \starMax^{-}(n, \psi)
      \\
    \starMax^{-}(n, \varphi \rightarrow \psi) &= \starMax^{+}(n, \varphi) \oplus \starMax^{-}(n, \psi)
      \\
    \starMax^{-}(n, \nec{\alpha}{\varphi}) &= 
      \sum_{\vec{a} \in \truncateAt{n}{\languageOf(\alpha)}} {\sum_{k \in \combinationsFor_{\alpha}(\vec{a})} \starMax^{-}(n - k, \varphi)}
        \\[-0.25em] &
        \hspace{7em} {} + 
      \sum_{\vec{a} \in \truncateAt{n}{\languageOf(\alpha)}} {\starMax^{-}({n - {\left| \vec{a} \right|}}, \varphi)}
  \end{align*}
  The function $\combinationsFor_{\alpha}$, parameterised by test-free programs $\alpha$, takes words (sequences of atomic programs) to sets of indices, and is defined as follows.
  \begin{align*}
    \combinationsFor_{a}(b) &= \emptyset \hspace{3cm} \text{for all atomic programs $a$ and $b$}
      \\
    \combinationsFor_{\choice{\alpha}{\beta}}(w) &= \combinationsFor_{\alpha}(w) \cup \combinationsFor_{\beta}(w)
      \\
    \combinationsFor_{\comp{\alpha}{\beta}}(w) &=
      \left\{ 
        k \left| 
          \begin{aligned}
            & \text{$w = w_1 \cdot w_2$, $w_1 \in \languageOf(\alpha)$, $w_2 \in \languageOf(\beta)$, and}
                \\
            & \hspace{3cm}
              \text{$k \in \combinationsFor_{\alpha}(w_1)$ or $k - \left| w_1 \right| \in \combinationsFor_{\beta}(w_2)$}
          \end{aligned}
          \right.
      \right\}
      \\
    \combinationsFor_{\iter{\alpha}}(w) &=
      \{ 0 \} \cup {}
      \left\{ 
        k \left| 
          \begin{aligned}
            & \text{$w = w_1 \cdot w_2$, ${\left| w_1 \right|} > 0$, $w_1 \in \languageOf(\alpha)$, and}
                \\
            & \hspace{1.85cm}
              \text{$k \in \combinationsFor_{\alpha}(w_1)$ or $k - \left| w_1 \right| \in \combinationsFor_{\iter{\alpha}}(w_2)$}
          \end{aligned}
          \right.
      \right\}
  \end{align*}
  We then define the $\starMax$ function on sequents, as follows.
  \begin{multline*}
    \starMax(\sequent{\Gamma}{\Delta}) 
      = \sum_{x : \varphi \in \Gamma} {\starMax^{-}(\pathMax(x, \sequent{\Gamma}{\Delta}), \varphi)}
            \\[-1em]
        {} + \sum_{x : \varphi \in \Delta} {\starMax^{+}(\pathMax(x, \sequent{\Gamma}{\Delta}), \varphi)}
  \end{multline*}
\end{definition}

Here again, we define two functions with different `polarities'.
For an initial sequent $\sequent{\Gamma}{\Delta}$, $\starMax^{+}(n, \varphi)$ gives the maximum number of consequent iterated formulas that will be generated along any path in an unwinding by an occurrence of $x : \varphi$ in $\Delta$, provided it can make use of a reachability path of length at most $n$.
Similarly, $\starMax^{-}$ gives the maximum number of consequent iterated formulas generated by $x : \varphi \in \Gamma$.
The reason we need a bound on the length of the available reachability path is precisely to reason about unfolding of antecedent iterated formulas.
Indeed, the only case in which this parameter is used in the definitions above is in that for an antecedent modal formula, $\starMax^{-}(n, \nec{\alpha}{\varphi})$.
Since the construction of an unwinding does not \emph{a priori} bound the number of times such a formula can be unfolded.
What actually prevents this process continuing \emph{ad infinitum} is the fact that the path of reachability, which is generated by unfolding consequent iterated formulas and decomposing the resulting program modalities, is bounded by the fact that unfolding of consequent formulas can happen at most once.
This principle can be seen at work in the proof of \cref{lem:FiniteUnwindings}.

The cases of the definition for conjunction, disjunction and implication should be relatively self-explanatory.
For example, because the ($\wedge$R) rule decomposes a conjunction into two separate branches, $\starMax^{+}$ takes the maximum of the values for the immediate sub-formulas.
For a disjunction, it adds the values because the ($\vee$R) rule decomposes the disjunction within the same branch.
Similarly for an implication $\varphi \rightarrow \psi$, the values are added as there is only a single premise for the ($\rightarrow$R) rule; notice the polarity switch for the recursion on the antecedent sub-formula $\varphi$.
The defintion of $\starMax^{-}$ for these cases is dual, since the left-hand proof rules are dual to the right-hand ones.

The most interesting cases of the definition are the ones for modalities.
We first describe the case for consequent modal formulas $\nec{\alpha}{\varphi}$.
If $\alpha$ is $\ast$-free, then the modality will simply be decomposed, and so the function simply recurses on the immediate sub-formula $\varphi$.
If $\alpha$ contains iterations, then the modality will be decomposed until an iteration is reached; however note that the decomposition will lead to only one iteration in each branch.
The ($\ast$R) rule will then be applied; it has two premises.
In the left-hand branch, the iteration is simply discarded, leaving the sub-formula $\varphi$.
In the right-hand branch, the iteration is unfolded and this process will begin again.
This eventually ends since iterated formulas will be unfolded exactly once, but we have that each path ends either with an iterated formula already unfolded or $\varphi$.
Thus, the number of iterated formulas generated along any of these paths is the maximum of $1$ (for the former case), or that of $\varphi$ (in the latter).

For antecedent modal formulas $\nec{\alpha}{\varphi}$, the case is slightly more complex.
These will be decomposed in potentially many ways according to the available paths of reachability, and so may lead to many occurrences of the sub-formula $\varphi$ along the same branch.
The definition of the case $\starMax^{-}(n, \nec{\alpha}{\varphi})$ thus makes use of the information about the maximum length $n$ of these available paths of reachability.
Any sequence of atomic programs (i.e.~word) that is included in (i.e.~is a model of) the program (i.e.~regular expression) $\alpha$, and which has length at most $n$ can lead to a new occurrence of $\varphi$ in the antecedent.
Therefore we sum over all such words.
The function $\combinationsFor_{\alpha}(\vec{a}_n)$ gives the set of all iteration `boundaries' for the sequence of atomic programs $\vec{a}$ with respect to the program $\alpha$.
That is, it gives the set of all indices $k < n$ such that an antecedent iterated formula $x : \nec{\alpha}{\varphi}$ can be unfolded to (essentially) $x : \nec{a_1}{\ldots\nec{a_k}{\nec{\iter{\beta}}{\varphi}}}$, for some program $\beta$, and thus decomposed to $y : \nec{\iter{\beta}}{\varphi}$ (provided there is a path of reachability from $x$ to $y$ via $a_1, \ldots, a_n$).
This means that the ($\square$L) rule will generate an antecedent formula $y : \varphi$.
Therefore, $\starMax^{-}(n, \alpha)$ takes the sum of $\starMax^{-}(n - k, \varphi)$ for all iteration boundaries $k$ in all words $\vec{a}$ of length up to $n$ that are in the language of the program $\alpha$.
Since $\combinationsFor_{\alpha}(\vec{a}_n)$ does not contain the index $n$ itself, $\starMax^{-}(n, \alpha)$ also includes in the sum $\starMax^{-}(n - {\left| \vec{a} \right|}, \varphi)$ for all such words $\vec{a}$.

To compute the value of $\starMax$ for a sequent, first the maximum length of the available reachability path for each label in the sequent is determined, using the $\pathMax$ function.
Then, using this information, we take the sum of the value of $\starMax$ (with the appropriate polarity) for each of the lablled formulas it contains.

For termination of proof-search, we need two properties.
The first is that the number of consequent iterated formulas in the open leaves of unwindings is bounded as a function of the initial sequent, namely $\starMax$.
The second is that the value of $\starMax$ does not increase during proof-search.
This means that the value of $\starMax$ for the initial sequent of the proof-search gives a bound on the size of the conclusions of all unwindings constructed during the search.
Thus, proof-search must terminate.

\begin{proposition}
\label{lem:Termination:BoundingFunction}
  For $\mathcal{D}$ be a capped unwinding for test-free, acyclic $\sequent{\Gamma}{\Delta}$, then $\left| \{ x : \varphi \in \Delta' \mid \text{$\varphi$ non-atomic} \} \right| \leq \starMax(\sequent{\Gamma}{\Delta})$ for all its open leaves $\sequent{\Gamma'}{\Delta'}$.
\end{proposition}

The second property requires some auxiliary monotonicity results.

\begin{lemma}[Path Monotonicity]
\label{lem:Monotonicity}
  \begin{enumerate}[label={\arabic*.},ref={{\thelemma}(\arabic*)},topsep={\smallskipamount}]
    \item
    \label[lemma]{lem:Monotonicity:PathMax:Left}
    $\pathMax(x, \sequent{\Gamma}{\Delta}) \leq \pathMax(x, \sequent{\Gamma, A}{\Delta})$.
    \item
    \label[lemma]{lem:Monotonicity:PathMax:Right}
    $\pathMax(x, \sequent{\Gamma}{\Delta}) \leq \pathMax(x, \sequent{\Gamma}{A, \Delta})$.
    \item
    \label[lemma]{lem:Monotonicity:StarMax:Formulas}
    $\starMax^{\pi}(m, \varphi) \leq \starMax^{\pi}(n, \varphi)$ for $\pi \in \{ +, - \}$ and $m < n$.
  \end{enumerate}
\end{lemma}
\begin{proof}
  The first two are straightforward.
  If $A$ is a relational atom, the result is immediate since these do not contribute to the value of $\pathMax$ for sequents.
  If $A$ is a labelled formula $y : \varphi$, then there are two cases to consider.
  For the case $x \neq y$ the result is immediate since the set of formulas from which we take the maximum does not change.
  If on the other hand $x = y$, the result is straightforward since we only add a formula to the set from which we take the maximum: if $\pathMax^{-}(\varphi)$ (resp.~$\pathMax^{+}(\varphi)$) is less than $\pathMax^{-}(\psi)$ (resp.~$\pathMax^{+}(\psi)$) for some $y : \psi \in \Gamma$ (resp.~$y : \psi \in \Delta$), then the value of $\pathMax(x, \cdot)$ for the two sequents are the same; otherwise it is greater for the larger sequent, as required.

  We prove the third property by induction on the structure of formulas.
  The only non-trivial case is for $\starMax^{-}(n, \nec{\alpha}{\varphi})$.
  In this case, the result follows straightforwardly from the induction and the fact that $\truncateAt{m}{\languageOf(\alpha)} \subseteq \truncateAt{n}{\languageOf(\alpha)}$ when $m < n$.
  \qed
\end{proof}



We now prove monotonicity of the $\starMax$ function for unwindings.
In fact, this holds more generally for finite (test-free) derivations of logical rules of any height.

\begin{proposition}[Monotonicity of {$\starMax$}]
\label{lem:Termination:BoundNonIncreasing}
  For $\mathcal{D}$ be a capped unwinding for test-free, acyclic $\sequent{\Gamma}{\Delta}$, then $\starMax(\sequent{\Gamma'}{\Delta'}) \leq \starMax(\sequent{\Gamma}{\Delta})$, for all its open leaves $\sequent{\Gamma'}{\Delta'}$.
\end{proposition}
\begin{proof}
  By induction on the structure of derivations.
  By transitivity of $\leq$, it suffices to show that for the conclusion $\sequent{\Gamma}{\Delta}$ and each premise $\sequent{\Gamma'}{\Delta'}$ of any proof rule, $\starMax(\sequent{\Gamma'}{\Delta'}) \leq \starMax(\sequent{\Gamma}{\Delta})$.
  Notice that we do not need to consider {\normalfont{(Cut)}} and {\normalfont{(Subst)}}, since these are not used in unwindings.
  \begin{description}[font={\normalfont}]
    \item[(WL), (WR):]
    If the principal formula is a relational atom the result is immediate since these do not contribute to the value of $\starMax$ for sequents.
    If the principal formula is a labelled formula $x : \varphi$ then the result follows by \cref{lem:Monotonicity}.
    \Cref{lem:Monotonicity:PathMax:Left,lem:Monotonicity:PathMax:Right} gives that the value of $\pathMax$ for every label is greater for the conclusion that it is for the premise.
    Then, for every labelled formula common to both the premise and concluse, \cref{lem:Monotonicity:StarMax:Formulas} gives that its value is greater in the conclusion; therefore, the sum for the conclusion is greater than the sum for the premise.
    The new labelled formula in the conclusion can only increase the value of the sum.
    \item[Classical connectives:]
    These cases follow relatively straightforwardly.
    We work through the left- and right-hand cases for conjunction as an illustration.
    \begin{description}[font={\normalfont}]
      \item[($\wedge$L):]
      Let $S$ denote the conclusion sequent $\sequent{x : \varphi \wedge \psi, \Gamma}{\Delta}$ and $S'$ denote the premise sequent $\sequent{x : \varphi, x : \psi, \Gamma}{\Delta}$.
      First of all note that, since $\pathMax^{-}(\varphi \wedge \psi) = \pathMax^{-}(\varphi) \oplus \pathMax^{-}(\psi)$, we have that $\pathMax(S) = \pathMax(S')$.
      Therefore, we need to show that the following inequality holds: $\starMax^{-}(n, \varphi) + \starMax^{-}(n, \psi) \leq \starMax^{-}(n, \varphi \wedge \psi)$, where $n = \pathMax(x, S) = \pathMax(x, S')$.
      This follows directly from \cref{def:StarMax}, which shows that they are in fact equal.
      \item[$(\wedge$R):] 
      Let $S$, $S'_1$, and $S'_2$ denote the conclusion sequent $\sequent{\Gamma}{\Delta, x : \varphi \wedge \psi}$, the left-hand premise sequent $\sequent{\Gamma}{\Delta, x : \varphi}$, and the right-hand premise sequent $\sequent{\Gamma}{\Delta, x : \psi}$, respectively.
      We first consider the result for $S'_1$.
      Here, since $\pathMax(\varphi \wedge \psi)$ is the maximum of $\pathMax(\varphi)$ and $\pathMax(\psi)$, we have only that $\pathMax(x, S'_1) \leq \pathMax(x, S)$.
      But then by \cref{lem:Monotonicity:StarMax:Formulas}, we have that the following holds for all $x : \varphi$ in $\Gamma$: $\starMax^{-}(\pathMax(x, S'_1), \varphi) \leq \starMax^{-}(\pathMax(x, S), \varphi)$; and that $\starMax^{+}(\pathMax(x, S'_1), \varphi) \leq \starMax^{+}(\pathMax(x, S), \varphi)$ for all $x : \varphi$ in $\Delta$.
      Thus it suffices to show that the following inequality also holds: $\starMax^{+}(\pathMax(x, S'_1), \varphi) \leq \starMax^{+}(\pathMax(x, S), \varphi \wedge \psi)$.
      By \cref{def:StarMax}, this holds if we also have the that following inequality holds: $\starMax^{+}(\pathMax(x, S'_1), \varphi) \leq \starMax^{+}(\pathMax(x, S), \varphi)$.
      But this follows directly from \cref{lem:Monotonicity:StarMax:Formulas}.
      The result for the right-hand premise $S'_2$ is shown symmetrically.
    \end{description}
%
  \rrnote{TO DO.} %
    \qed
  \end{description}
\end{proof}

\CutfreeRegularCompleteness*
\begin{proof}
  The sequent $\sequent{}{x : \varphi}$ is test-free and acyclic.
  We start with the open derivation consisting of the single open leaf $\sequent{}{x : \varphi}$.
  A proof can be found by iteratively building capped unwindings for the currently open leaves.
  By \cref{lem:UnwindingProperties:ValidLeaves}, all open leaves at each stage are valid.
  Moreover, since $\varphi$ is test-free, so are all sequents occurring during proof-search, which contain only formulas in the Fischer-Ladner closure of $\varphi$.
  By \cref{lem:UnwindingsPreserveAcyclicity}, each open leaf at each stage is also acyclic.
  By \cref{lem:FiniteUnwindings}, each unwinding is finite, and so can be built in finite time.
  \Cref{lem:Termination:BoundingFunction,lem:Termination:BoundNonIncreasing} ensure that the number of labelled iterated formulas in each open leaf encountered in proof-search is bounded by $\starMax(\varphi)$.
  Moreover, \cref{lem:FiniteUnwindings} ensures that the bound on the number of iterated formulas is also a bound on the number of labels occurring in the open leaves.
  This means that the size of the sequent in each open leaf is bounded.
  That is, the number of distinct open leaves encountered during proof-search is bounded, modulo label substitution.
  Therefore, during proof-search, if any open leaf is a substitution instance of a previously encountered sequent we can apply a substitution, convert it to a bud and assign the previously encountered node as its companion.
  We must eventually end up with a pre-proof after finite time.
  Due to the sub-formula property of all rules except ($\ast$L) and ($\ast$R), for any infinite path through the pre-proof there must be a path of ancestry following some consequent iterated (sub-)formula.
  Since all unwindings are finite, the infinite path must cycle through some unwinding an infinite number of times.
  Moreover, since sequents are finite, the path must go through some labelled formula $x : \varphi$ in the conclusion of this unwinding an infinite number of times.
  By the construction of unwindings, $\varphi$ is an iterated formula $\nec{\iter{\alpha}}{\psi}$.
  We can therefore form a trace from this formula starting with $\tau = (x, \emptysequence, \alpha, \psi)$ that follows this path.
  Then by \cref{lem:UnwindingProperties:Traces}, each segment of the trace between occurrences of this formula (of which there are infinitely many) is progressing.
  Thus, the pre-proof satisfies the global trace condition.
  \qed
\end{proof}

\end{document}

%% file: fig_pdlrules.tex
\begin{figure}[t!]
  {\smaller
  \begin{gather*}
    \text{(Ax):} \quad
    \begin{prooftree}
      \phantom{\sequent{A}{A}}
        \justifies
      \sequent{A}{A}
    \end{prooftree}
      \qquad
    \text{($\bot$):} \quad
    \begin{prooftree}
      \phantom{\sequent{x : \bot}{}}
        \justifies
      \sequent{x : \bot}{}
    \end{prooftree}
      \qquad
    \text{(WL):} \quad
    \begin{prooftree}
      \sequent{\Gamma}{\Delta}
        \justifies
      \sequent{A, \Gamma}{\Delta}
    \end{prooftree}
      \qquad
    \text{(WR):} \quad
    \begin{prooftree}
      \sequent{\Gamma}{\Delta}
        \justifies
      \sequent{\Gamma}{\Delta, A}
    \end{prooftree}
      \\[0.75em]
    \text{($\wedge$L):} \quad
    \begin{prooftree}
      \sequent{{x : \varphi}, {x : \psi}, \Gamma}{\Delta}
        \justifies
      \sequent{{x : \varphi \wedge \psi}, \Gamma}{\Delta}
    \end{prooftree}
      \qquad
    \text{($\wedge$R):} \quad
    \begin{prooftree}
      \sequent{\Gamma}{\Delta, {x : \varphi}}
        \quad
      \sequent{\Gamma}{\Delta, {x : \psi}}
        \justifies
      \sequent{\Gamma}{\Delta, {x : \varphi \wedge \psi}}
    \end{prooftree}
      \\[0.75em]
    \text{($\vee$L):} \quad
    \begin{prooftree}
      \sequent{{x : \varphi}, \Gamma}{\Delta}
        \quad
      \sequent{{x : \psi}, \Gamma}{\Delta}
        \justifies
      \sequent{{x : \varphi \vee \psi}, \Gamma}{\Delta}
    \end{prooftree}
      \qquad
    \text{($\vee$R):} \quad
    \begin{prooftree}
      \sequent{\Gamma}{\Delta, {x : \varphi}, {x : \psi}}
        \justifies
      \sequent{\Gamma}{\Delta, {x : \varphi \vee \psi}}
    \end{prooftree}
      \\[0.75em]
    \text{($\rightarrow$L):} \quad
    \begin{prooftree}
      \sequent{\Gamma}{\Delta, {x : \varphi}}
        \quad
      \sequent{{x : \psi}, \Gamma}{\Delta}
        \justifies
      \sequent{{x : \varphi \rightarrow \psi}, \Gamma}{\Delta}
    \end{prooftree}
      \qquad
    \text{($\rightarrow$R):} \quad
    \begin{prooftree}
      \sequent{{x : \varphi}, \Gamma}{\Delta, {x : \psi}}
        \justifies
      \sequent{\Gamma}{\Delta, {x : \varphi \rightarrow \psi}}
    \end{prooftree}
      \\[0.75em]
    \text{($\square$L):} \quad
    \begin{prooftree}
      \sequent{y : \varphi, \Gamma}{\Delta}
        \justifies
      \sequent{{x : \nec{a}{\varphi}}, {x \relatedBy{a} y}, \Gamma}{\Delta}
    \end{prooftree}
      \qquad
    \text{($\square$R):} \quad
    \begin{prooftree}
      \sequent{{x \relatedBy{a} y}, \Gamma}{\Delta, {y : \varphi}}
        \justifies
      \sequent{\Gamma}{\Delta, {x : \nec{a}{\varphi}}}
        \using {\text{\smaller{($y$ fresh)}}}
    \end{prooftree}
      \\[0.75em]
    \text{($\compSymb$L):} \quad
    \begin{prooftree}
      \sequent{{x : \nec{\alpha}{\nec{\beta}{\varphi}}}, \Gamma}{\Delta}
        \justifies
      \sequent{{x : \nec{\comp{\alpha}{\beta}}{\varphi}}, \Gamma}{\Delta}
    \end{prooftree}
      \qquad
    \text{($\compSymb$R):} \quad
    \begin{prooftree}
      \sequent{\Gamma}{\Delta, {x : \nec{\alpha}{\nec{\beta}{\varphi}}}}
        \justifies
      \sequent{\Gamma}{\Delta, {x : \nec{\comp{\alpha}{\beta}}{\varphi}}}
    \end{prooftree}
      \\[0.75em]
    \text{($\choiceSymb$L):} \quad
    \begin{prooftree}
      \sequent{{x : \nec{\alpha}{\varphi}}, {x : \nec{\beta}{\varphi}}, \Gamma}{\Delta}
        \justifies
      \sequent{{x : \nec{\choice{\alpha}{\beta}}{\varphi}}, \Gamma}{\Delta}
    \end{prooftree}
      \qquad
    \text{($\choiceSymb$R):} \quad
    \begin{prooftree}
      \sequent{\Gamma}{\Delta, {x : \nec{\alpha}{\varphi}}}
        \quad
      \sequent{\Gamma}{\Delta, {x : \nec{\beta}{\varphi}}}
        \justifies
      \sequent{\Gamma}{\Delta, {x : \nec{\choice{\alpha}{\beta}}{\varphi}}}
    \end{prooftree}
      \\[0.75em]
    \text{($?$L):} \quad
    \begin{prooftree}
      \sequent{\Gamma}{\Delta, {x : \varphi}}
        \quad
      \sequent{{x : \psi}, \Gamma}{\Delta}
        \justifies
      \sequent{{x : \nec{\test{\varphi}}{\psi}}, \Gamma}{\Delta}
    \end{prooftree}
      \qquad
    \text{($?$R):} \quad
    \begin{prooftree}
      \sequent{{x : \varphi}, \Gamma}{\Delta, {x : \psi}}
        \justifies
      \sequent{\Gamma}{\Delta, {x : \nec{\test{\varphi}}{\psi}}}
    \end{prooftree}
      \\[0.75em]
    \text{($\ast$L):} \quad
    \begin{prooftree}
      \sequent{{x : \varphi}, {x : \nec{\alpha}{\nec{\iter{\alpha}}{\varphi}}}, \Gamma}{\Delta}
        \justifies
      \sequent{{x : \nec{\iter{\alpha}}{\varphi}}, \Gamma}{\Delta}
    \end{prooftree}
      \qquad
    \text{($\ast$R):} \quad
    \begin{prooftree}
      \sequent{\Gamma}{\Delta, {x : \varphi}}
        \quad
      \sequent{\Gamma}{\Delta, {x : \nec{\alpha}{\nec{\iter{\alpha}}{\varphi}}}}
        \justifies
      \sequent{\Gamma}{\Delta, {x : \nec{\iter{\alpha}}{\varphi}}}
    \end{prooftree}
      \\[0.75em]
    \text{(Subst):} \quad
    \begin{prooftree}
      \sequent{\Gamma}{\Delta}
        \justifies
      \sequent{\Gamma \subst{x}{y}}{\Delta \subst{x}{y}}
    \end{prooftree}
      \qquad
    \text{(Cut):} \quad
    \begin{prooftree}
      \sequent{\Gamma}{\Delta, A}
        \quad
      \sequent{A, \Sigma}{\Pi}
        \justifies
      \sequent{\Gamma, \Sigma}{\Delta, \Pi}
    \end{prooftree}
  \end{gather*}}
  \caption{Inference rules of {\InfinitarySystem}}
  \label{fig:PDLRules}
\end{figure}